\documentclass[a4paper,11pt]{article}
\pdfoutput=1 

\usepackage{jheppub} 

\usepackage[T1]{fontenc} 
\usepackage{subcaption}
\usepackage{verbatim}
\usepackage{mathrsfs}
\usepackage{physics}
\usepackage{amsthm}
\usepackage{mathtools}
\usepackage{enumerate}

\renewcommand{\hat}{\widehat}
\renewcommand{\bar}{\overline}

\newcommand{\comps}{\mathbb{C}}

\newcommand{\A}{\mathcal{A}}
\newcommand{\M}{\mathcal{M}}
\renewcommand{\H}{\mathcal{H}}

\newcommand{\im}{\operatorname{im}}

\newcommand{\B}{\mathcal{B}}

\newcommand{\Z}{\mathcal{Z}}

\newcommand{\DGNS}{D_{\text{GNS}}}
\newtheorem{theorem}{Theorem}
\newtheorem{corollary}[theorem]{Corollary}
\newtheorem{lemma}[theorem]{Lemma}

\theoremstyle{definition}

\numberwithin{theorem}{section}

\title{Continuum canonical purifications}
\author{Jonathan Sorce}
\affiliation{Princeton Gravity Initiative}
\abstract{
We construct and characterize canonical purifications for general algebraic states, extending prior constructions by Woronowicz and by Dutta/Faulkner to general quantum field theories.
Given a quantum state on a $*$-algebra, the canonical purification is a state on a ``doubled'' algebra that admits an interpretation in terms of CRT reflection.
This interpretation holds for all quantum theories, even in the absence of gravity.
We then identify conditions under which canonical purifications are ``pure'' in the technical sense, compute their modular conjugations, and relate them to GNS and natural-cone purifications in certain settings.
In an appendix, we develop a general theory of von Neumann algebras generated by unbounded $*$-algebras.
In a forthcoming paper with Caminiti and Capeccia, we provide an application of this general theory to the problem of excitability in quantum field theory.
}

\begin{document}
\maketitle

\section{Introduction}

In most applications of quantum information theory, one starts with a physically interesting pure state and studies the entanglement structure of its subsystems.
When studying such a subsystem, it can be helpful to remove the privileged status of the original pure state by studying the full set of allowable purifications.
This gives rise, for example, to Uhlmann's characterization of the quantum fidelity \cite{Uhlmann:fidelity} and to the entanglement of purification \cite{Terhal:EOP}.

For a density matrix $\rho$ on a finite-dimensional Hilbert space $\H$, a \textit{canonical} purification is provided by thinking of the operator $\sqrt{\rho}$ as a state in the Hilbert space $\H\otimes \H^*$, where $\H^*$ is the dual space to $\H$.
This construction has been used in quantum information theory since at least \cite{Hughston:classification}, but has become newly popular in recent years due to applications in high energy physics.
In \cite{Dutta:canonical}, building on earlier work in \cite{Maldacena:TFD, EW:1, EW:2}, it was argued that for a geometric state in a holographic quantum field theory, the canonical purification has a dual gravitational description in terms of a ``CRT-sewing'' of two copies of the entanglement wedge.
This inspired many interesting investigations into canonical purifications in quantum gravity and in quantum field theory more generally; for an incomplete list, see \cite{Bousso:shocks, Marolf:sewing, Bueno:fermions, Bueno:scalars, Zou:universal, Engelhardt:simple, Hayden:Markov, Akers:random-1, Engelhardt:evaporating, Akers:Page, Akers:random-2, Dutta:fermions, Parrikar:shock, Engelhardt:transfer}.

Despite these successes, most investigations mentioned above use the language of density matrices and regularized, tensor-factorizing Hilbert spaces.
This limits our ability to apply canonical purifications to problems that are most readily addressed in the continuum.
One step in this direction was made in \cite[appendix A]{Dutta:canonical} by relating canonical purifications of density matrices to the GNS construction, and there is no essential subtlety in generalizing that appendix to continuous, faithful states on C$^*$ algebras.
But there are many important quantum states that are not faithful, and many important settings in which quantum theory is not intrinsically described by C$^*$ algebras; these include quantum field theory in general curved backgrounds \cite{Hollands:review} and potentially closed-universe quantum gravity \cite{Witten:background}.

The purpose of this paper and its companion paper \cite{Caminiti:paper-2} is to elaborate the general theory of canonical purifications, and to provide a concrete application.
In fact, part of this goal was already accomplished by Woronowicz in 1972 \cite{Woronowicz:purification}, though Woronowicz's paper has been underappreciated in the literature.
In \cite{Woronowicz:purification}, Woronowicz constructed a general theory of canonical purifications for so-called ``factor'' states on C$^*$ algebras.
This is partially more general than the construction in \cite{Dutta:canonical}, as it relaxes the assumption of faithfulness; however it requires making an additional assumption of factoriality.
It also remains firmly in the realm of C$^*$ algebras, which --- as discussed in the preceding paragraph --- are not appropriate for describing general, interacting quantum field theories.

In order to make the theory of canonical purifications applicable in complete generality, the present paper develops a general theory of canonical purifications for quantum states on $*$-algebras,\footnote{We review in section \ref{sec:general-algebras} why this is an important general framework for thinking about physics.} and is essentially an expansion of \cite{Woronowicz:purification, Dutta:canonical} to a broader class of physical settings.
We also develop a general theory of what it means to construct a von Neumann algebra out of an abstract $*$-algebra, which does not seem to have appeared in the literature before.
The companion paper \cite{Caminiti:paper-2}, co-authored with Caminiti and Capeccia, makes a connection between canonical purifications and the ``symplectic purifications'' that have been used in algebraic free field theory since \cite{Powers:quasi}.
That work then uses the general theory of canonical purifications to simplify and generalize the main theorems of \cite{Powers:quasi, Araki:quasi, VanDaele:quasi, Araki-Yamagami}.

Concretely, given any $*$-algebra $\A_0$ and any algebraic quantum state $\omega,$\footnote{We do not actually consider \textit{any} state $\omega$; we assume the GNS representation has a countable basis, which should be the case for algebras and states describing local quantum physics.} we construct and characterize an algebraic state $\hat{\omega}$ that has the interpretation of ``two copies of $\omega$ sewed together after time-reversal.''
See figure \ref{fig:time-reversal}.
The state $\hat{\omega}$ is defined on the $*$-algebra $\A_0 \otimes \A_0^{\text{op}},$ where the ``opposite algebra'' $\A_{0}^{\text{op}}$ has the same elements as $\A_0$ but a reversed multiplication rule.
The reversal of multiplication ordering plays the role of time reversal, since in canonical quantization, reversing the order of multiplication is the same as changing the sign of momentum.
The general formula, which will be explained in much more detail in section \ref{sec:intuition}, is\footnote{For the special case of a C$^*$ algebra, this map is the one considered by Woronowicz in \cite{Woronowicz:purification}.}
\begin{equation} \label{eq:intro-main-formula}
	\hat{\omega}(a \otimes b)
		= \langle a^\dagger \omega | J_{\omega} |b^{\dagger} \omega\rangle.
\end{equation}
The right-hand side is evaluated in the GNS representation of $\omega$, and $J_{\omega}$ is the corresponding modular conjugation.
From this formula, one can show that $\hat{\omega}$ reduces to $\omega$ on either subsystem:
\begin{equation}
	\hat{\omega}(a \otimes 1) = \hat{\omega}(1 \otimes a) = \omega(a).
\end{equation}

\begin{figure}
	\centering
	\includegraphics[scale=1.4]{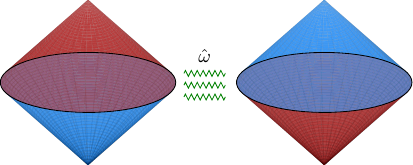}
	\caption{The canonical purification is an entangled algebraic state, $\hat{\omega},$ defined on two copies of the same degrees of freedom, with one copy subject to time reversal.
	In the language of the text, the left causal diamond contains the degrees of freedom $\A_0$, while the right causal diamond contains the degrees of freedom $\A_0^{\text{op}}$.}
	\label{fig:time-reversal}
\end{figure}

We review in section \ref{sec:intuition} how equation \eqref{eq:intro-main-formula} reproduces the density-matrix ket $|\sqrt{\rho}\rangle$ in the finite-dimensional setting.
In the case that $\omega$ is a faithful, continuous state on a von Neumann algebra, we also demonstrate the existence of a canonical isomorphism between the GNS spaces $\H_{\hat{\omega}}$ and $\H_{\omega},$ which essentially uses the modular conjugation to identify the ``purifying system'' in $\H_{\omega}$ with the time-reversed algebra $\A_0^{\text{op}}.$
This conceptual shift provides the key utility of the canonical purification: one starts with the GNS purification, in which the purifying degrees of freedom do not come with any intrinsic structure, then collapses these purifying degrees of freedom into a time-reversed copy of the original degrees of freedom.
See figure \ref{fig:GNS-crunch}.
Physical arguments involving the canonical purification are easier to pursue due to this additional structure.

\begin{figure}
	\centering
	\includegraphics[scale=1.4]{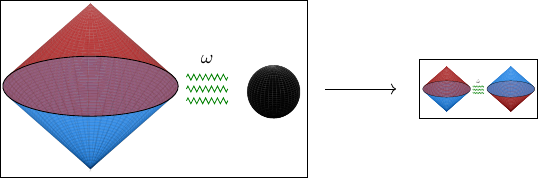}
	\caption{The GNS purification introduces entangled degrees of freedom with no intrinsic structure; this is represented by the left side of the above figure.
		In the case that $\omega$ is faithful, the canonical purification arises by using the modular conjugation to identify the new degrees of freedom with a second copy of the original ones.}
	\label{fig:GNS-crunch}
\end{figure}

We then extend the canonical purification formula beyond the regime of faithful, continuous states on von Neumann algebras. 
The relaxation of the faithfulness condition makes the canonical purification distinct from the GNS purification, resulting in an embedding $\H_{\omega} \hookrightarrow \H_{\hat{\omega}}.$
For example, a pure state $|\psi\rangle \in \mathbb{C}^d$ is not faithful on the matrix algebra $\operatorname{Mat}_{d \times d}(\mathbb{C}),$ and in this setting the GNS construction returns the state $|\psi\rangle,$ while the canonical purification returns $|\psi\rangle \otimes |\psi\rangle.$
The relaxation of continuity and boundedness assumptions allows canonical purification technology to be applied to general interacting quantum field theories, where the fundamental degrees of freedom cannot a priori be organized into von Neumann algebras.
We prove that in the general unbounded, non-faithful setting, $\hat{\omega}$ is still always a positive linear functional (i.e., it is an algebraic state), and we show that the time-reversal is necessary for positivity to hold.

Next, we investigate the conditions under which $\hat{\omega}$ is pure.
We establish purity in the case that $\omega$ is ``factorial,'' which is the setting in which the von Neumann algebra associated with $\A_0$ has a trivial center.
There is an interesting caveat to this discussion, which is that for unbounded $*$-algebras there are two senses in which $\hat{\omega}$ could be pure --- a ``strong purity,'' which is the statement that $\hat{\omega}$ cannot be nontrivially decomposed as a sum of two algebraic states, and a ``weak purity,'' which is the statement that the GNS representation of $\hat{\omega}$ introduces no nontrivial purifying system.
The difference between these is explained in appendix \ref{app:purity-theorem}; the result of this paper is that $\hat{\omega}$ is weakly pure in the factorial case, though it may not always be pure in the strong sense.
In the general setting, we show that $\omega$ is ``centrally pure,'' a term we introduce for the property $\A_{\hat{\omega}}' \subseteq \A_{\hat{\omega}}$ --- this means that while the GNS representation may introduce new, commuting degrees of freedom, they can all be obtained as limits of the degrees of freedom we start with.\footnote{Lauritz van Luijk has informed us that using disintegration theory, one can show that this is equivalent to the statement that $\A_{\hat{\omega}}$ is a direct integral of type I factors with trivial multiplicity, meaning that each factor appearing in the direct integral is of the form $\B(\H)$ for some Hilbert space $\H$.}

Finally, we compute the modular conjugation operator associated with $\hat{\omega}$ in the case where $|\omega\rangle$ is separating for the von Neumann algebra $\A_{\omega}$ generated by $\A_0.$
The formula is
\begin{equation} \label{eq:intro-J}
	J_{\hat{\omega}} (a \otimes b) |\hat{\omega}\rangle = (b^* \otimes a^*) |\hat{\omega}\rangle.
\end{equation}
We use this to establish a connection between the ``time-reflection'' canonical purification and the natural-cone purification that has occasionally been called canonical.
Concretely, whenever two algebraic states $\omega_1$ and $\omega_2$ admit a unitary equivalence between their GNS spaces $\H_{\hat{\omega}_1}$ and $\H_{\hat{\omega}_2}$, we show that the representative of $\hat{\omega}_1$ in $\H_{\hat{\omega}_2}$ can be brought into the natural cone of $|\hat{\omega}_2\rangle$ by acting with a unitary in the center of the von Neumann algebra generated by $(\A_0 \otimes 1)$.

Much of the technical material in the paper has been relegated to two appendices, and we have tried to make the main text readable to a broad audience.
However it will be assumed that the reader is familiar with basic notions concerning von Neumann algebras and modular theory, such as the double commutant theorem, the Tomita operator, etc.
Readers who need a refresher are encouraged to consult the reviews in \cite{Witten:notes, Sorce:types, Sorce:modular}.

The plan of the paper follows.

In section \ref{sec:intuition}, we review the structure of the canonical purification in finite-dimensional systems, rewrite everything in terms of modular quantities, and use this to define the canonical purification for any algebraic state on a general $*$-algebra.
Expanding on \cite{Dutta:canonical}, we explain the interpretation of the canonical purification in the separating, bounded case as a procedure by which the purifying GNS degrees of freedom are canonically identified with the time-reversal of the original system.
We also give an overview of the work in \cite{Woronowicz:purification} and the issues that arise when generalizing this discussion to states on a $*$-algebra, which are then resolved by the rest of the paper.

In section \ref{sec:faithful}, we establish the basic theory of canonical purifications in the case that $\omega$ is separating but the representation of $\A_0$ is unbounded: we show that $\hat{\omega}$ is positive as a functional, establish the isomorphism of GNS spaces $\H_{\omega} \cong \H_{\hat{\omega}},$ derive equation \eqref{eq:intro-J}, and show that when $\omega$ is factorial, the functional $\hat{\omega}$ is pure in the sense that the GNS space has no nontrivial purifying system.
We also characterize the structure of the purifying system in the non-factorial case.

In section \ref{sec:non-separating}, we extend the theory to the case of non-separating states.
We provide no computation of the modular conjugation in this setting, but we still establish positivity of $\hat{\omega}$, an embedding of GNS spaces $\H_{\omega} \hookrightarrow \H_{\hat{\omega}},$ and purity in the factorial case. 

Finally, in section \ref{sec:comparison}, we explain the connection between canonical and natural-cone purifications in the case of a unitary equivalence of sectors.

All Hilbert spaces, in particular the Hilbert spaces associated to GNS representations, are assumed to admit a countable orthonormal basis.

\section{Finite dimensions and beyond}
\label{sec:intuition}

\subsection{Finite-dimensional formulas}
\label{sec:finite-dimension}

Let $\H$ be a finite-dimensional Hilbert space, with $\B(\H)$ the algebra of linear operators acting on $\H$.
Consider a density matrix $\rho$ with diagonal form
\begin{equation}
	\rho = \sum_{j} p_j \ketbra{j}.
\end{equation}
This density matrix defines an expectation value functional on $\B(\H)$ via the map
\begin{equation}
	\omega(a)
		= \tr(\rho a).
\end{equation}
A purification of $\rho$ can be defined by taking a second copy of $\H$ and writing
\begin{equation}
	|\sqrt{\rho}\rangle = \sum_{j} \sqrt{p_j} |j\rangle |j\rangle.
\end{equation}
This formula involves a choice of identification between the first and second copies of $\H$, so a slightly more ``canonical'' formulation involves thinking of $|\sqrt{\rho}\rangle$ as a state on the Hilbert space $\H \otimes \H^*,$ with $\H^*$ the dual space to $\H$.
Since $\H \otimes \H^*$ is the space of operators acting on $\H$, the ``ket'' $|\sqrt{\rho}\rangle$ is the same as the operator $\sqrt{\rho}.$

Every operator that acts on $\H$ can also be made to act on $\H^*$ via the formula
\begin{equation} \label{eq:right-action}
	a \cdot \langle \psi |
		= \langle \psi| a = \langle a^{\dagger} \psi|.
\end{equation}
Note that equation \eqref{eq:right-action} furnishes a \textit{right} action of $\B(\H)$ on $\H^*.$
Equivalently, it is a left action of $\B(\H)^\text{op}$, which is the algebra with the same elements as $\B(\H)$, but with a reversed multiplication rule.
So the space of all linear operators on $\H\otimes \H^*$ is naturally identified with $\B(\H) \otimes \B(\H)^{\text{op}},$ and if $O$ is an operator on $\H$ interpreted as a vector in $\H \otimes \H^*$, then we have
\begin{equation}
	(a \otimes b)|O\rangle
		= |a O b\rangle.
\end{equation}
The ket $|\sqrt{\rho}\rangle$ gives a set of expectation values for $\B(\H) \otimes \B(\H)^{\text{op}}$ via the linear functional
\begin{equation} \label{eq:rho-hat}
	\hat{\omega}(a \otimes b)
		= \langle \sqrt{\rho}| (a \otimes b) |\sqrt{\rho}\rangle.
\end{equation}

Our goal is to generalize equation \eqref{eq:rho-hat} to states that are not represented by density matrices, and to $*$-algebras that are not of the form $\B(\H)$ for some Hilbert space $\H$.
To do this, we will review from \cite[appendix A]{Dutta:canonical} the connection between equation \eqref{eq:rho-hat} and the modular quantities associated with the GNS representation of $\rho.$
We will then produce a general definition using the modular version of equation \eqref{eq:rho-hat}.

A complete explanation of the GNS construction can be found in appendix \ref{app:GNS}.
The idea is to endow the vector space $\B(\H)$ with an inner product
\begin{equation}
	\langle a | b \rangle_{\text{GNS}}
		= \omega(a^{\dagger} b),
\end{equation}
then to take a quotient by the null states of this inner product.
The resulting Hilbert space is called the ``GNS space'' $\H_{\omega}$.
The state associated with the identity operator is usually just called $|\omega\rangle$, and the state that was labeled $|a\rangle_{\text{GNS}}$ is now written as $|a \omega\rangle = a |\omega\rangle.$
There is a natural map from the GNS space to $\H\otimes \H^*$ defined by
\begin{equation}
	V : a |\omega\rangle \mapsto (a \otimes 1) |\sqrt{\rho}\rangle = |a \sqrt{\rho}\rangle.
\end{equation}
This is an isometry, but it is not always a unitary map.
If $\rho$ is not full rank, then the GNS construction involves a nontrivial quotient, and the isometry $V$ maps onto a subspace of $\H\otimes \H^*.$
For example, if one starts with a pure density matrix $\rho = \ketbra{\psi},$ then one has $V \H_{\omega} = \H \otimes \bra{\psi}.$
In general, if we let $P$ be the projection onto the support of $\rho,$ then the image of $V$ is $\H \otimes P \H^*.$

Despite the potential non-isomorphism between $\H_{\omega}$ and $\H \otimes \H^*$, one can still always express the expectation value functional \eqref{eq:rho-hat} purely in terms of operators acting on the GNS representation $\H_{\omega}$.
The reason is that the modular operator $\Delta_{\omega}$ on the GNS representation transforms nicely under $V$.
If one defines the inverse density matrix $\rho^{-1}$ to vanish on the complement of the support of $\rho,$ then one has
\begin{equation} \label{eq:V-transport-of-Delta}
	V \Delta_{\omega} V^{\dagger}
		= \rho \otimes \rho^{-1}.
\end{equation}
For a calculation of this formula in the full-rank case, see \cite[section 4.1]{Witten:notes}.
In the current setting of a non-full-rank density matrix, the modular operator is still defined as in appendix \ref{app:modular}, and from the formula
\begin{equation}
	\langle a \omega | \Delta_{\omega} | b \omega \rangle
		= \langle b^{\dagger} \omega | P | a^{\dagger} \omega\rangle,
\end{equation}
one can deduce equation \eqref{eq:V-transport-of-Delta}.
To compare this to an expression involving $|\sqrt{\rho}\rangle,$ one writes
\begin{equation}
	\langle \sqrt{\rho} | (a \otimes b) |\sqrt{\rho}\rangle
	= \tr(\sqrt{\rho} a \sqrt{\rho} b)
	= \langle \sqrt{\rho} | (a \otimes 1) (\sqrt{\rho} \otimes \sqrt{\rho}^{-1}) (b \otimes 1) |\sqrt{\rho}\rangle,
\end{equation}
then pulls this expression back through the action of the isometry $V$ to obtain
\begin{equation}
	\hat{\omega}(a \otimes b)
		= \langle \sqrt{\rho} | (a \otimes b) |\sqrt{\rho}\rangle
		= \langle a^{\dagger} \omega| \Delta_{\omega}^{1/2} b |\omega\rangle.
\end{equation}
Using the identities of appendix \ref{app:modular}, one can rewrite this in a slightly more symmetric form as
\begin{equation} \label{eq:density-matrix-modular-formula}
	\hat{\omega}(a \otimes b)
		= \langle a^{\dagger} \omega| J_{\omega} | b^{\dagger} \omega\rangle.
\end{equation}

Equation \eqref{eq:density-matrix-modular-formula} tells us how to compute the functional $\hat{\omega}$ using the modular data of the GNS representation $\H_{\omega}$.
In the case of a quantum state described by a density matrix, we had an independent definition of canonical purification, and \eqref{eq:density-matrix-modular-formula} was derived.
In the general, $*$-algebraic setting described further in section \ref{sec:general-algebras}, we can begin with an algebraic state $\omega$ on the $*$-algebra $\A_0,$ and use equation \eqref{eq:density-matrix-modular-formula} as the \textit{definition} of a linear functional $\hat{\omega}$ on the $*$-algebra $\A_0 \otimes \A_0^{\text{op}}.$
For this to make sense, we will have to establish that this actually defines a reflection-positive state in the general $*$-algebraic setting, which we will do in later sections.
First, however, we will expand on \cite[appendix A]{Dutta:canonical} to establish an isomorphism between $\H_{\hat{\omega}}$ and $\H_{\omega}$ in the case where $\A_0$ is a von Neumann algebra and $\omega$ is continuous and faithful; the intuition is that in this setting, one can canonically identify the purifying degrees of freedom in $\H_{\omega}$ with a second copy of the degrees of freedom in $\A_0.$

\subsection{Time-reversal and GNS}
\label{sec:time-reversal-GNS}

If $\A_0$ is a von Neumann algebra and $\omega$ is a continuous (``ultraweakly continuous'') and faithful state, then the representation of $\A_0$ on the GNS space $\H_{\omega}$ is an isomorphism of von Neumann algebras.\footnote{We will return to this fact in section \ref{sec:separating-extension}. For a proof, see for example \cite[proposition 5.18]{Stratila:book}.}
In other words, the minimal von Neumann algebra $\A_{\omega}$ generated by $\A_0$ is identical to $\A_0$ --- no new operators appear by taking limits.

The GNS space $\H_{\omega}$ carries a natural action of $\A_0^{\text{op}}$ via the formula
\begin{equation}
	a^{\text{op}} (b |\omega\rangle) = b a |\omega\rangle.
\end{equation}
In general, however, this operator $a^{\text{op}}$ need not be bounded.
A better behaved action of $\A_0^{\text{op}}$ on $\H_{\omega}$ can be constructed using the modular conjugation $J_{\omega}.$
In the present setting, where $\omega$ is faithful, the vector $|\omega\rangle$ is both cyclic \textit{and separating} for the von Neumann algebra $\A_{\omega}.$
Consequently, the modular conjugation provides an antilinear isomorphism between $\A_{\omega}$ and the commutant algebra $\A_{\omega}'$ via the formula
\begin{equation}
	J_{\omega} \A_{\omega} J_{\omega} = \A_{\omega}'.
\end{equation}
This can be made into a linear isomorphism by introducing a dagger:
\begin{equation}
	a \mapsto J_{\omega} a^{\dagger} J_{\omega}.
\end{equation}
The formula
\begin{equation}
	a^{\text{op}} |\psi\rangle = J_{\omega} a^{\dagger} J_{\omega} |\psi\rangle
\end{equation}
furnishes a bounded action of $\A_{0}^{\text{op}}$ on $\H_{\omega}$, and identifies this algebra with $\A_{\omega}'.$

As indicated in the introduction, we can think of the commutant algebra $\A_{\omega}'$ as encoding the ``purifying degrees of freedom'' introduced by the GNS construction.
If one starts with an algebraic state $\omega$ that is mixed on the von Neumann algebra $\A_0$, then one ends up introducing a nontrivial commutant algebra in the GNS space.\footnote{See appendix \ref{app:purity-theorem} for a discussion of what it means for an abstract algebraic state to be mixed, and for an explanation of the relation to the GNS commutant.}
There is no intrinsic structure to this commutant, but if $\omega$ is faithful, then the modular conjugation $J_{\omega}$ can be used to identify $\A_{\omega}'$ with a second copy of $\A_0$.
To make this identification linear, one takes an adjoint and replaces $\A_0$ with $\A_0^{\text{op}}$.
This induces a linear functional $\hat{\omega}$ on $\A_0 \otimes \A_0^{\text{op}}$ via the formula
\begin{equation}
	\hat{\omega}(a \otimes b)
		\equiv \langle \omega | a J_{\omega} b^{\dagger} J_{\omega} |\omega\rangle
		= \langle a^{\dagger} \omega | J_{\omega} | b^{\dagger} \omega\rangle,
\end{equation}
generalizing the finite-dimensional formula \eqref{eq:density-matrix-modular-formula}.
It is not so hard to show that $\hat{\omega}$ is a positive linear functional, and that the GNS spaces $\H_{\hat{\omega}}$ and $\H_{\omega}$ are related by the unitary map
\begin{equation}
	(a \otimes b) |\hat{\omega}\rangle \mapsto a J_{\omega} b^{\dagger} J_{\omega} |\omega\rangle.
\end{equation}
We will perform this calculation in detail, in a more general setting, in section \ref{sec:non-separating}.

We therefore see that when we start with a continuous, faithful state on a von Neumann algebra, the GNS space possesses ``hidden structure'' that can be expressed in terms of the original degrees of freedom and a time-reversed copy --- see again figure \ref{fig:GNS-crunch}.
When $\omega$ is non-faithful, the identity $J_{\omega} \A_{\omega} J_{\omega} = \A_{\omega}'$ no longer holds --- see appendix \ref{app:modular} --- and the purifying degrees of freedom in the GNS space provide only a partial copy of $\A_0.$
We will discuss this further in section \ref{sec:non-separating}.
For now, we will focus on the subtleties that arise when $\A_0$ is a general $*$-algebra, in which case no continuity properties can be imposed on $\omega$ and the GNS operators are unbounded.

\subsection{The need for generalization}
\label{sec:general-algebras}

Especially in the last few years, a broad community of high-energy theorists has come to appreciate the utility of describing subsystems of a quantum field theory in terms of von Neumann algebras.
However, there is an important and lesser-known setting in which algebraic reasoning is still essential, but where von Neumann algebras are not a part of the fundamental structure.
This is the abstract description of a theory without reference to a particular choice of Hilbert space.
For example, in free field theory in Minkowski spacetime, the Hilbert spaces describing thermal physics at different values of the temperature are inequivalent due to infrared divergences, but should still be thought of as ``sectors'' of a single underlying theory.
The situation is worse in curved spacetime, since in the absence of a time translation symmetry there is no preferred class of thermal Hilbert spaces to discuss.
Instead, as explained for example in \cite{Hollands:axioms}, a general Lorentzian quantum field theory is most naturally formulated in terms of a $*$-algebra, which is an abstract collection of symbols representing the quantum fields, subject to certain algebraic relations and potentially additional structure like an operator product expansion.
One arrives at a particular Hilbert space sector of this theory by choosing an algebraic state --- i.e., a self-consistent set of correlation functions --- and performing the GNS construction.

In a general, interacting quantum field theory, this structure of an abstract $*$-algebra cannot be avoided.
Moreover, it has recently been proposed \cite{Witten:background} that a similar structure is required to describe the experience of a semiclassical observer in closed-universe quantum gravity.
Mathematical subtleties abound due to the fact that a $*$-algebra $\A_0,$ represented on a GNS space $\H_{\omega},$ is generally represented by unbounded operators instead of bounded ones.
To apply the von Neumann algebraic tools that are used to study quantum information theory in the continuum, one must learn how to generate well behaved algebras of bounded operators out of the $*$-algebraic fields.
Procedures of this type have been studied in \cite{Driessler:unbounded-to-vN, Buchholz:unbounded-to-vN}, and our favored approach is explained in appendix \ref{app:star-algebras}.
While intuitive reasoning at the level of von Neumann algebras is sufficient for reaching a conceptual understanding, our point of view is that there are important problems where the language of unbounded operators is both natural and physically transparent, and that developing information-theoretic tools that apply to this general setting is a promising avenue for progress.
The companion paper \cite{Caminiti:paper-2} will study an example of one such problem.

Concretely, a $*$-algebra $\A_0$ is a collection of symbols that can be multiplied and added to one another, that can be multiplied by complex scalars, and that possess an ``adjoint operation'' that implements complex conjugation on multiples of the identity.
The adjoint satisfies the formula
\begin{equation}
	(ab)^* = b^* a^*.
\end{equation}
An algebraic state $\omega$ is a linear functional from $\A_0$ to $\mathbb{C}$ satisfying the ``reflection positivity'' or simply ``positivity'' condition $\omega(a^* a) \geq 0.$
One should think of $\A_0$ as being ``an abstract algebra of fields,'' and one should think of $\omega$ as being ``a set of correlation functions.''
Given a choice of $\omega$, one can obtain a GNS space $\H_{\omega}$ carrying a representation of $\A_0.$
The GNS space possesses a special state $|\omega\rangle$, and a dense set of states obtained by acting on $|\omega\rangle$ with elements of $\A_0.$
This defines, for every $a \in \A_0,$ an operator that knows how to act on every state $b |\omega\rangle$ for $b$ in $\A_0$; but this operator will generally be unbounded, and will not be definable on every state in $\H_{\omega}.$
The overlaps in the GNS space reproduce the correlation functions of the algebraic state via the formula
\begin{equation}
	\langle \omega | a |\omega\rangle = \omega(a).
\end{equation}

In general, the operators coming from $\A_0$ cannot belong to a von Neumann algebra, since a von Neumann algebra is by definition a set of bounded operators.
To proceed, one can define the von Neumann algebra $\A_{\omega}$ to be the smallest von Neumann algebra such that every operator coming from $\A_0$ is ``affiliated.''
For more detail, see appendix \ref{app:star-algebras}.
The chief difficulty arises from the fact that the Hilbert space adjoint $a^{\dagger}$ is generally a distinct operator from $a^*$; this comes from ``domain issues'' intrinsic to unbounded operators, and the difference between $a^{\dagger}$ and $a^*$ can be severe.
Casual readers can largely ignore this subtlety, but for the interested reader we emphasize that due to domain issues, the operator $a + a^*$ is generally not self-adjoint and therefore not diagonalizable.
This causes an obstruction to defining a von Neumann algebra, since one cannot simply define $\A_{\omega}$ to be generated by all bounded truncations of operators like $a + a^*$.
Instead, for each $a \in \A_0$ one examines the polar decomposition $a = V_a |a|,$ and defines $\A_{\omega}$ to be the smallest von Neumann algebra containing all such partial isometries $V_a$ and all bounded functions of the self-adjoint operators $|a|.$

The ``opposite'' $*$-algebra $\A_0^{\text{op}}$ has the same elements as $\A_0$ and a reversed multiplication rule.
The tensor product $\A_0 \otimes \A_0^{\text{op}}$ is the set of all finite sums of multilinear symbols $a \otimes b.$
As a generalization of the finite-dimensional formula \eqref{eq:density-matrix-modular-formula}, we define the linear functional
\begin{equation}
	\hat{\omega} : \A_0 \otimes \A_0^{\text{op}} \to \comps
\end{equation}
via the formula
\begin{equation} \label{eq:defining-equation}
	\hat{\omega}(a \otimes b)
		\equiv \langle a^* \omega | J_{\omega} | b^* \omega\rangle.
\end{equation}
Since $a^*$ and $b^*$ are unbounded operators, one cannot freely rewrite this as
\begin{equation}
	\hat{\omega}(a \otimes b)
	\equiv \langle \omega | a J_{\omega} b^* J_{\omega} \omega\rangle
\end{equation}
without being careful about domain issues.
This means that proofs of the basic properties of $\hat{\omega}$ --- its positivity, and the isomorphism between $\H_{\hat{\omega}}$ and $\H_{\omega}$ in the case that $|\omega\rangle$ is separating for $\A_{\omega}$ --- require some care.
In the next section, we address these issues in the separating case.
We note that time-reversal is necessary for equation \eqref{eq:defining-equation} to define a positive linear functional, since if one interprets equation \eqref{eq:defining-equation} as a functional on $\A_0 \otimes \A_0$, then one can find simple finite-dimensional examples where it is non-positive; for a brief elaboration, see appendix \ref{app:time-reversal}.

We comment, as in the introduction, that the above equations for $\hat{\omega}$ as a state on $\A_0 \otimes \A_0^{\text{op}}$ have appeared previously in the work of Woronowicz \cite{Woronowicz:purification} in the special case that $\A_0$ is a C$^*$ algebra and $\omega$ is a continuous, factorial state.
That paper did not make the connection with the finite-dimensional canonical purification $|\sqrt{\rho}\rangle,$ and the work of \cite{Dutta:canonical} did not make a connection with the general theory of C$^*$ algebras.
The present work unifies these two perspectives on canonical purification and generalizes them both.

\section{Canonical purifications of separating states}
\label{sec:faithful}

In the preceding section, we considered a generic $*$-algebra $\A_0$ with an algebraic state $\omega$, and defined an algebraic state $\hat{\omega}$ on $\A_0 \otimes \A_0^{\text{op}}$ by employing modular quantities from the GNS representation of $\omega$.
In this section, we study the basic structure of $\hat{\omega}$ in the case that $\omega$ is separating.
We show that $\hat{\omega}$ is a positive functional, meaning it has a well defined GNS representation; we demonstrate a unitary isomorphism between the GNS spaces $\H_{\hat{\omega}}$ and $\H_{\omega}$; we prove the modular conjugation formula \eqref{eq:intro-J}; and we demonstrate purity of $\hat{\omega}$ in the case that $\omega$ is a factorial state.

\subsection{Basic algebraic properties}

For every $a \in \A_0,$ one obtains an unbounded operator on the GNS space $\H_{\omega}.$
Details are discussed in appendix \ref{app:star-algebras}.
This operator $a$ has a domain that is not all of $\H_{\omega},$ but that includes all vectors of the form $b|\omega\rangle$ with $b \in \A_0.$
The operator $a$ is ``closable,'' meaning it is well behaved with respect to limits.
By abuse of notation we will use the symbol $a$ for the closure.
For distinct $a, b \in \A_0,$ the domains of these closed operators will be different, but they share the fundamental ``GNS domain''
\begin{equation}
	\DGNS
		= \{c |\omega\rangle \,|\, c \in \A_0\}.
\end{equation}

In appendix \ref{app:star-algebras}, we explain how to construct a von Neumann algebra $\A_{\omega}$ from $\A_0,$ and we prove that $|\omega\rangle$ is always cyclic for this von Neumann algebra.
We say that $\omega$ is a \textit{separating} state if $|\omega\rangle$ is separating for $\A_{\omega}.$
If $\omega$ were a continuous state on a von Neumann algebra, then this would be equivalent to the faithfulness condition that $\omega(a^* a) = 0$ implies $a=0$.
However, in the general, unbounded setting, these are not the same.
The case where $\omega$ is separating is the one where the properties of $\hat{\omega}$ are easiest to study.

In this setting, we expect to have an unbounded generalization of the statement that $J_{\omega}$ maps $\A_{\omega}$ to $\A_{\omega}'.$
We want an equation of the schematic form
\begin{equation} \label{eq:bad-commutator}
	a J_{\omega} b J_{\omega} = J_{\omega} b J_{\omega} a, \qquad a, b \in \A_0.
\end{equation}
This equation cannot literally be true as written due to domain issues coming from the unboundedness of $a$ and $b.$
However, one can show that for any $a, b \in \A_0,$ the vector $J_{\omega} b |\omega\rangle$ is in the domain of $a,$ and one has the vector identity
\begin{equation} \label{eq:domain-J-switching}
	a J_{\omega} b |\omega\rangle
		= J_{\omega} b J_{\omega} a |\omega\rangle.
\end{equation}
This is the unbounded version of equation \eqref{eq:bad-commutator}; for a proof, see appendix \ref{app:commutant-domain}.

\subsection{Positivity}
\label{sec:separating-positivity}

Now we will show that when $\omega$ is separating, the functional $\hat{\omega}$ is positive, i.e., it is actually an algebraic state.
We must demonstrate the inequality
\begin{equation} \label{eq:sec-3-positivity}
	\sum_{j, k} \hat{\omega}((a_j^* \otimes b_j^*) (a_k \otimes b_k)) \geq 0
\end{equation}
for any finite sum.
This is easy to do using equation \eqref{eq:domain-J-switching}.
We start with the calculation
\begin{align}
	\begin{split}
	\sum_{j, k} \hat{\omega}((a_j^* \otimes b_j^*) (a_k \otimes b_k))
		& = \sum_{j, k} \hat{\omega}(a_j^* a_k \otimes b_k b_j^*) \\
		& = \sum_{j, k} \langle a_k^* a_j \omega | J_{\omega} b_j b_k^* \omega\rangle.
	\end{split}
\end{align}
Using equation \eqref{eq:domain-J-switching}, and using the result from appendix \ref{app:GNS} that $a_k$ is equal on its domain to the adjoint of $a_k^*$, we can move $a_k^*$ from the bra to the ket, obtaining
\begin{align}
	\begin{split}
		\sum_{j, k} \hat{\omega}((a_j^* \otimes b_j^*) (a_k \otimes b_k))
		& = \sum_{j, k} \langle a_j \omega | J_{\omega} b_j b_k^* J_{\omega} a_k \omega\rangle.
	\end{split}
\end{align}
We then move $J_{\omega}$ from the ket to the bra using the fact that $J_{\omega}$ is an antilinear operator satisfying $J_{\omega}^{\dagger} = J_{\omega},$ which gives
\begin{align}
	\begin{split}
		\sum_{j, k} \hat{\omega}((a_j^* \otimes b_j^*) (a_k \otimes b_k))
		& = \sum_{j, k} \langle b_j b_k^* J_{\omega} a_k \omega | J_{\omega} a_j \omega\rangle.
	\end{split}
\end{align}
Once more we invoke equation \eqref{eq:domain-J-switching} to move $b_j$ from the bra to the ket, obtaining
\begin{align}
	\begin{split}
		\sum_{j, k} \hat{\omega}((a_j^* \otimes b_j^*) (a_k \otimes b_k))
		& = \sum_{j, k} \langle J_{\omega} a_k J_{\omega} b_k^* \omega | J_{\omega} a_j J_{\omega} b_j^*  \omega\rangle.
	\end{split}
\end{align}
Finally, we use antiunitary of $J_{\omega}$ to simplify this to
\begin{align}
	\begin{split}
		\sum_{j, k} \hat{\omega}((a_j^* \otimes b_j^*) (a_k \otimes b_k))
		& = \sum_{j, k} \langle a_j J_{\omega} b_j^* \omega | a_k J_{\omega} b_k^*  \omega\rangle.
	\end{split}
\end{align}
This is the norm-squared of a vector, so it is nonnegative.

\subsection{GNS isomorphism}
\label{sec:GNS-isomorphism}

Now that we know $\hat{\omega}$ is a positive linear functional in the separating case, we can study its GNS representation $\H_{\hat{\omega}}.$
In particular, we can consider the map from $\H_{\hat{\omega}}$ to $\H_{\omega}$ given by
\begin{equation} \label{eq:separating-unitary}
	(a \otimes b) |\hat{\omega}\rangle
		\mapsto a J_{\omega} b^* |\omega\rangle.
\end{equation}
The calculation of the preceding subsection shows that this is an isometry.
It clearly has dense range, since its range would be dense in $\H_{\omega}$ even if we set $b = 1.$
Equation \eqref{eq:separating-unitary} therefore extends to a unitary isomorphism from $\H_{\hat{\omega}}$ to $\H_{\omega}.$
This extends the discussion of section \ref{sec:time-reversal-GNS} to non-continuous states on general $*$-algebras.

\subsection{Computing the modular conjugation}
\label{sec:modular-properties}

In the preceding subsection, we demonstrated a unitary equivalence between the GNS spaces  $\H_{\hat{\omega}}$ and $\H_{\omega}$ in the case that $|\omega\rangle$ is separating for $\A_{\omega}.$
Under this unitary equivalence, the von Neumann algebra generated by $\A_0 \otimes 1$ is mapped to the von Neumann algebra $\A_{\omega}$, and the state $|\hat{\omega}\rangle$ is mapped to $|\omega\rangle.$
So under the assumptions of the present section, the canonical purification $|\hat{\omega}\rangle$ is cyclic and separating for the von Neumann algebra generated by $(\A_0 \otimes 1).$
What is its modular conjugation, $J_{\hat{\omega}}$?

Under the unitary equivalence from equation \eqref{eq:separating-unitary}, $J_{\hat{\omega}}$ will map to $J_{\omega}.$
So, with ``$\sim$'' denoting unitary equivalence, one has
\begin{equation}
	J_{\hat{\omega}} (a \otimes b) |\hat{\omega}\rangle
	\sim J_{\omega} a J_{\omega} b^* |\omega\rangle.
\end{equation}
Using equation \eqref{eq:domain-J-switching}, one computes
\begin{equation}
	J_{\hat{\omega}} (a \otimes b) |\hat{\omega}\rangle
	\sim b^* J_{\omega} a |\omega\rangle
	\sim (b^* \otimes a^*) |\hat{\omega}\rangle.
\end{equation}
This gives
\begin{equation}
	J_{\hat{\omega}} (a \otimes b) |\hat{\omega}\rangle
	= (b^* \otimes a^*) |\hat{\omega}\rangle,
\end{equation}
as claimed in the introduction.

\subsection{Purity in the factorial case}
\label{sec:separating-purity}

We now wish to understand the conditions under which $\hat{\omega}$ is pure.
As explained in appendix \ref{app:purity-theorem}, in the general unbounded setting there are two senses in which $\hat{\omega}$ may be called pure.
We will study the weaker of these, and say that $\hat{\omega}$ is pure if its GNS purification contains no nontrivial purifying system; i.e., if the only operators on $\H_{\hat{\omega}}$ that commute\footnote{In the parlance of appendix \ref{app:star-algebras}, we mean ``commute strongly.''} with $\A_0 \otimes \A_0^{\text{op}}$ and its adjoints are complex multiples of the identity.
Note that in the bounded case, as explained in appendix \ref{app:purity-theorem}, this is equivalent to the ``strong purity'' statement that $\hat{\omega}$ cannot be written as a nontrivial mixture of other positive functionals.
We will now prove that $\hat{\omega}$ is pure in the weak sense if and only if the state $\omega$ is factorial, meaning that the von Neumann algebra $\A_{\omega}$ has trivial center.

We will establish the relationship between purity and factoriality as a corollary to a slightly stronger, more generally useful statement, broken out as its own lemma below.
For the moment we continue to assume that $\omega$ is separating, though an analogous lemma will be established in the non-separating case in section \ref{sec:separating-purity}.

\begin{lemma}
	On $\H_{\hat{\omega}}$, we will call $(\A_0 \otimes 1)_{\hat{\omega}}, (1 \otimes \A_0^{\text{op}})_{\hat{\omega}},$ and $(\A_0 \otimes \A_0^{\text{op}})_{\hat{\omega}}$ the minimal von Neumann algebras generated by the $*$-algebras in parentheses.
	Under the unitary equivalence of equation \eqref{eq:separating-unitary}, these are mapped respectively to $\A_{\omega}, \A_{\omega}',$ and $\A_{\omega} \vee \A_{\omega}',$ where $\vee$ is the algebraic union.
\end{lemma}
\begin{proof}
	Let $U$ denote the unitary from equation \eqref{eq:separating-unitary}.
	Direct computation yields
	\begin{equation}
		U (a \otimes 1) U^{\dagger}
			= a
	\end{equation}
	and
	\begin{equation}
		U (1 \otimes a) U^{\dagger}
			= J_{\omega} a^* J_{\omega}.
	\end{equation}
	The first of these clearly establishes the equality
	\begin{equation}
		U (\A_0 \otimes 1)_{\hat{\omega}} U^{\dagger}
			= \A_{\omega}.
	\end{equation}
	The second is a little more subtle to treat in the unbounded case due to the technical nature of the definitions.
	In appendix \ref{app:vN}, we define $\A_0'$	to be the set of bounded operators on $\H_{\omega}$ that commute ``strongly'' with $\A_0,$ i.e., the set of all operators $L'$ that preserve the domain of each $a \in \A_0$ and commute with $a$ on that domain.
	In that same appendix, we established the formula
	\begin{equation}
		\A_{\omega}'
			= \A_0' \cap (\A_0')^{\dagger}.
	\end{equation}
	But how do we see that this is the same as the von Neumann algebra generated by the unbounded operators $J_{\omega} a J_{\omega}$ with $a \in \A_0$?
	
	The key observation, which is very easy to check, is that $L$ commutes strongly with $J_{\omega} a J_{\omega}$ if and only if $J_{\omega} L J_{\omega}$ commutes strongly with $a.$
	This establishes the equality of sets
	\begin{equation}
		(J_{\omega} \A_0 J_{\omega})' \cap ((J_{\omega} \A_0 J_{\omega})')^{\dagger}
			= J_{\omega} \A_0' J_{\omega} \cap (J_{\omega} \A_0' J_{\omega})^{\dagger} 
			= J_{\omega} \A_{\omega}' J_{\omega}
			= \A_{\omega}.
	\end{equation}
	Taking an additional commutant gives equality between $\A_{\omega}'$ and the von Neumann algebra generated by the unbounded operators $J_{\omega} \A_0 J_{\omega}.$
	This establishes the claim that $U$ takes $(1 \otimes \A_0^{\text{op}})_{\hat{\omega}}$ to $\A_{\omega}'.$
	
	To finish the proof, it remains to establish the identity
	\begin{equation}
		U (\A_0 \otimes \A_0^{\text{op}})_{\hat{\omega}} U^{\dagger}
			= \A_{\omega} \vee \A_{\omega}'.
	\end{equation}
	Given what we have already shown, we may equivalently show
	\begin{equation}
		(\A_0 \otimes \A_0^{\text{op}})_{\hat{\omega}}
		= (\A_0 \otimes 1)_{\hat{\omega}} \vee (1 \otimes \A_0^{\text{op}})_{\hat{\omega}}.
	\end{equation}
	It will be easier to do this in the equivalent commutant form,
	\begin{equation}
		(\A_0 \otimes \A_0^{\text{op}})_{\hat{\omega}}'
		= (\A_0 \otimes 1)_{\hat{\omega}}' \cap (1 \otimes \A_0^{\text{op}})_{\hat{\omega}}'.
	\end{equation}
	
	In the notation we used above for the commutants of sets of unbounded operators, we may rewrite our desired equality as
	\begin{equation} \label{eq:cap-lemma}
		(\A_0 \otimes \A_0^{\text{op}})' \cap [(\A_0 \otimes \A_0^{\text{op}})']^{\dagger}
		= (\A_0 \otimes 1)' \cap [(\A_0 \otimes 1)']^{\dagger} \cap (1 \otimes \A_0^{\text{op}})' \cap [(1 \otimes \A_0^{\text{op}})']^{\dagger}.
	\end{equation}
	Since commutation is a linear condition, it is easy to see that an operator commutes with $\A_0 \otimes 1$ and $1 \otimes \A_0^{\text{op}}$ if and only if it commutes with $\A_0 \otimes \A_0^{\text{op}}$; this establishes equality in equation \eqref{eq:cap-lemma}, which completes the proof.
\end{proof}

\begin{corollary}
	In the factorial case, we have
	\begin{equation}
		(\A_0 \otimes \A_0^{\text{op}})_{\hat{\omega}}'
		= \comps.
	\end{equation}
\end{corollary}
\begin{proof}
	By the lemma, it suffices to show that the commutant of $\A_{\omega} \vee \A_{\omega}'$ on $\H_{\omega}$ is trivial.
	But the commutant of $\A_{\omega} \vee \A_{\omega}'$ is the center $\A_{\omega} \cap \A_{\omega}'.$
	When $\omega$ is factorial, this center is trivial.
\end{proof}

\section{Canonical purifications of non-separating states}
\label{sec:non-separating}

\textit{Note:} The unbounded versions of the arguments in this section are significantly more complicated than the bounded ones, but the manipulations of non-separating states for bounded settings are interesting in their own right.
To accommodate readers who may want to understand these points without delving into unbounded operators, each subsection below has been structured so that bounded considerations are frontloaded; casual readers can skip to the next subsection as soon as unbounded issues are mentioned.

\subsection{The general strategy}
\label{sec:separating-extension}

A strategy we will employ multiple times in what follows involves extending a non-separating state to a separating one.
This trick is widely used in the literature on operator algebras, and is closely related to the theory of ``standard'' von Neumann algebras \cite{Araki:natural-1, Connes:standard, Haagerup:standard}.

Suppose that $|\omega\rangle$ is not separating for the von Neumann algebra $\A_{\omega}.$
This means that the subspace $\A'_{\omega}|\omega\rangle$ is not dense in the GNS space $\H_{\omega}.$
We denote the projector onto the closure of this subspace by $P_{\omega}.$
Let $|j\rangle$ be an orthonormal basis for the support of $(1-P_{\omega}),$ and define a new functional $\phi$ on $\A_\omega$ by\footnote{This is the appropriate formula for when there are an infinite number of basis elements $|j\rangle$; if there are only finitely many, then the sequence $1/2^j$ should be replaced by any finite, positive sequence that sums to one.} 
\begin{equation}
	\phi(L)
		= \frac{1}{2} \left[ \langle \omega | L |\omega\rangle + \sum_{j=1}^{\infty} \frac{1}{2^j} \langle j | L |j\rangle \right].
\end{equation}
This is easily seen to be properly normalized, and it satisfies the positivity condition
\begin{equation}
	\phi(L^{\dagger} L)
	\equiv \frac{1}{2} \left[ \langle L \omega | L \omega\rangle + \sum_{j=1}^{\infty} \frac{1}{2^j} \langle L j | L j \rangle \right]
	\geq 0.
\end{equation}
In fact, for the inequality to be saturated, one must have $L |\omega\rangle = 0$ and $L (1 - P_{\omega}) = 0.$
But $L|\omega\rangle = 0$ implies $L \A_{\omega}' |\omega\rangle = 0,$ hence $L P_{\omega} = 0,$ hence $L = 0.$
Combining our observations so far, we learn that $\phi(L^{\dagger} L)$ vanishes if and only if we have $L=0$, so $\phi$ is a faithful state on the von Neumann algebra $\A_{\omega}.$

Now let us consider the GNS space $\H_{\phi}.$
This is a space built from the algebra $\A_{\omega},$ so $\A_{\omega}$ has a natural action on $\H_{\phi}.$
The operator topologies on $\H_{\phi}$ are distinct from those on $\H_{\omega}$, and naively one might expect that $\A_{\omega}$ is not a von Neumann algebra on $\H_{\phi}$; that instead it must be ``completed'' into a larger algebra $\A_{\phi}.$
However, standard tricks in the theory of von Neumann algebras --- see \cite[proposition 5.18]{Stratila:book} --- tell us that the version of $\A_{\omega}$ that acts on $\H_{\phi}$ is already a von Neumann algebra.
To help us keep track of which space we are working in, we will call this algebra $\A_{\phi},$ even though its elements are in one-to-one correspondence with those of $\A_{\omega}.$
The utility of introducing $\phi$ is that the vector $|\phi\rangle$ is cyclic and separating for $\A_{\phi}.$

As explained in appendix \ref{app:cyclicity}, the state $|\omega\rangle$ is cyclic for $\A_{\omega}.$
Consequently, one can define a map from $\H_{\omega}$ to $\H_{\phi}$ via the formula\footnote{It makes sense to act with $P_{\omega}$ in the GNS space $\H_{\phi}$ because $P_{\omega}$ is an element of $\A_{\omega}$; see appendix \ref{app:modular}.}
\begin{equation} \label{eq:sec-3-V}
	V : L |\omega\rangle \mapsto \sqrt{2} L P_{\omega} |\phi\rangle,
\end{equation}
and this is easily checked to be an isometry.
However, $V$ is generally not unitary, since its image will not be dense in $\H_{\phi}.$
We can identify its image via an application of the centralizer theorem, which is a classic theorem in Tomita-Takesaki theory that we review in appendix \ref{app:centralizer}.
For any $L \in \A_{\phi}$, one has
\begin{equation}
	\langle \phi | L P_{\omega} |\phi\rangle = \frac{1}{2} \langle \omega | L | \omega \rangle = \langle \phi | P_{\omega} L |\phi\rangle.
\end{equation}
So using the results of appendix \ref{app:centralizer}, one has
\begin{equation}
	J_{\phi} P_{\omega} |\phi\rangle = \Delta_{\phi}^{1/2} P_{\omega} |\phi\rangle = P_{\omega} |\phi\rangle.
\end{equation}
Consequently, if we define
\begin{equation}
	P_{\phi}' \equiv J_{\phi} P_{\omega} J_{\phi},
\end{equation}
then we have
\begin{equation}
	\sqrt{2} L P_{\omega}|\phi\rangle = \sqrt{2} P_{\phi}' L |\phi\rangle.
\end{equation}
States of the form $L|\phi\rangle$ are dense in $\H_{\phi},$ so we see that the image of $V$ is the support of the projector $P_{\phi}'.$

Just as in the finite-dimensional case from section \ref{sec:finite-dimension}, the modular operator $\Delta_{\omega}$ on $\H_{\omega}$ is closely related to the modular operator $\Delta_{\phi}$ on $\H_{\phi}.$
To see this, we use the defining formula for the Tomita operator $S_{\omega}$ (see appendix \ref{app:modular}), which is
\begin{equation}
	S_{\omega} L |\omega\rangle
		= P_{\omega} L^{\dagger} |\omega\rangle
\end{equation}
If we act on the left with $V$, we obtain
\begin{equation}
	V S_{\omega} L |\omega\rangle
	= \sqrt{2} P_{\omega} L^{\dagger} P_{\omega} |\phi\rangle.
\end{equation}
We can then rewrite this as 
\begin{equation}
	V S_{\omega} L |\omega\rangle
	= \sqrt{2} S_{\phi} P_{\omega} L P_{\omega} |\phi\rangle
	= S_{\phi} P_{\omega} V L |\omega\rangle.
\end{equation}
So $V S_{\omega}$ is equal on its domain to $S_{\phi} P_{\omega} V$, and multiplying by adjoints gives the useful formula
\begin{equation} \label{eq:separating-mod-op}
	\Delta_{\omega}
		= V^{\dagger} P_{\omega} \Delta_{\phi} P_{\omega} V.
\end{equation}
From this we can also derive a formula for $J_{\omega}$.
Given $L_1, L_2 \in \A_{\omega},$ we have
\begin{align} 
	\begin{split}
	\langle L_1 \omega | J_{\omega} |L_2 \omega\rangle
		& = \langle L_1 \omega | \Delta_{\omega}^{1/2} |L_2^{\dagger} \omega \rangle \\
		& = \langle L_1 \omega | V^{\dagger} P_{\omega} \Delta_{\phi}^{1/2} P_{\omega} V |L_2^{\dagger} \omega \rangle \\
		& = 2 \langle P_{\omega} L_1 P_{\omega} \phi| \Delta_{\phi}^{1/2} |P_{\omega} L_2^{\dagger} P_{\omega} \phi\rangle \\
		& = 2 \langle P_{\omega} L_1 P_{\omega} \phi| J_{\phi} |P_{\omega} L_2 P_{\omega} \phi\rangle \\
		& = \langle L_1 \omega | V^{\dagger} P_{\omega} J_{\phi} P_{\omega} V |L_2 \omega\rangle,
	\end{split}
\end{align}
from which we can deduce
\begin{equation} \label{eq:separating-mod-conj}
	J_{\omega} = V^{\dagger} P_{\omega} J_{\phi} P_{\omega} V.
\end{equation}

The utility of equation \eqref{eq:separating-mod-conj} is that it gives us a simple formula relating the canonical purification $\hat{\omega}$ to certain expectation values of $|\phi\rangle.$
In the bounded case, where elements of $\A_0$ can be interpreted as elements of $\A_{\omega},$ one has
\begin{align} \label{eq:separating-ext-pen}
	\begin{split}
	\hat{\omega}(a \otimes b)
		& = \langle a^* \omega | J_{\omega} | b^*\omega\rangle \\
		& = \langle a^* \omega | V^{\dagger} P_{\omega} J_{\phi} P_{\omega} V | b^*\omega\rangle \\
		& = 2 \langle P_{\omega} a^* P_{\omega} \phi | J_{\phi} | P_{\omega} b^* P_{\omega} \phi\rangle.
	\end{split}
\end{align}
A slight simplification is afforded by using our earlier observation that $J_{\phi} P_{\omega} J_{\phi} = P_{\phi}'$ is an operator in $\A_{\phi}',$ and that it satisfies $P_{\phi}'|\phi\rangle = P_{\omega} |\phi\rangle.$
Using these identities, together with $J_{\phi}^2 = 1,$ one can rewrite the above as
\begin{align} \label{eq:separating-ext-eq}
	\begin{split}
		\hat{\omega}(a \otimes b)
		& = 2 \langle a^* P_{\omega} \phi | J_{\phi} | b^* P_{\omega} \phi\rangle \\
		& = 2\,\hat{\phi}(P_{\omega} a \otimes P_{\omega} b).
	\end{split}
\end{align}

In the unbounded case, one cannot perform all of the manipulations of the preceding paragraph.
The operator $a \in \A_0$ is not an element of $\A_{\omega},$ so it does not have an intrinsic action on $\H_{\phi},$ and we may not pass directly from line 2 to 3 of equation \eqref{eq:separating-ext-pen}.
We must first define an action of $\A_0$ on $\H_{\phi}$.
This is not so hard to do --- as explained in appendix \ref{app:star-algebras}, every element $a \in \A_0$ can be approximated on its domain by a sequence $a_n$ of bounded operators in $\A_{\omega}$ obtained via the polar decomposition.
Each of these operators $a_n$ acts on $\H_{\phi},$ so we simply define $a$ on $\H_{\phi}$ as their limit: we say $|\psi\rangle$ in $\H_{\phi}$ is in the domain of $a$ if the sequence $a_n |\psi\rangle$ converges.

Clearly $a$ has in its domain the vector $P_{\omega} |\phi\rangle,$ since we have
\begin{equation}
	\lim_n a_n P_{\omega}|\phi\rangle
		= \frac{1}{\sqrt{2}} \lim_n V a_n |\omega\rangle
		= \frac{1}{\sqrt{2}} V a |\omega\rangle.
\end{equation}
In fact, for any $b \in \A_0,$ the vector $b P_{\omega}|\phi\rangle$ is in the domain of $a$; for we have
\begin{equation}
	\lim_n a_n b P_{\omega} |\phi\rangle
		= \lim_n \lim_m a_n b_m P_{\omega} |\phi\rangle
		= \frac{1}{\sqrt{2}} \lim_n \lim_m V a_n b_m |\omega\rangle = \frac{1}{\sqrt{2}} V a b |\omega\rangle.
\end{equation}
From these limits we may deduce the formula
\begin{equation}
	a b P_{\omega} |\phi\rangle = \frac{1}{\sqrt{2}} V a b |\omega\rangle.
\end{equation}
Similar logic tells us that for any bounded $L' \in \A_{\phi}',$ the vectors $L' b P_{\omega}|\phi\rangle$ are also in the domain of $a,$ with action
\begin{equation}
	a L' b P_{\omega}|\phi\rangle
		= \frac{1}{\sqrt{2}} L' V a b |\omega\rangle.
\end{equation}
This allows every $a \in \A_0$ to be defined on the set $\A_{\phi}' V \A_0 |\omega\rangle$, which is a dense subspace of $\H_{\phi}$ (see footnote for proof),\footnote{\label{foot:density}What we want to show is that $V \A_0 |\omega\rangle$ is a cyclic subspace for $\A_{\phi}'$. This is the same as saying that it is a separating subspace for $\A_{\phi}.$ For $L \in \A_{\phi},$ if we have $L V a |\omega\rangle = 0$ for every $a \in \A_0,$  then this is the same as saying that $L$ vanishes on the image of $V$. Equivalently, because $\A_{\omega} |\omega\rangle$ is dense in $\H_{\omega},$ we have $L V T |\omega\rangle = 0$ for every $T \in \A_{\omega}.$ But this gives
\begin{equation}
	0 = L V T|\omega\rangle = L T P_{\omega} |\phi\rangle = V L T |\omega\rangle,
\end{equation}
and because $V$ is an isometry this implies $L T |\omega\rangle = 0$ for every $T \in \A_{\omega},$ hence $L = 0.$}
so every $a \in \A_0$ is densely defined on $\H_{\phi}.$
Moreover, it is straightforward to show that on this common domain, one has $a^{\dagger} = a^*$ --- this follows from the manipulation
\begin{align}
	\begin{split}
	\langle L_1' V b_1 \omega | a | L_2' V b_2\omega\rangle
		& = \lim_n \langle L_1' V b_1 \omega | a_n | L_2' V b_2\omega\rangle \\
		& = \lim_n \langle L_1' a_n^{\dagger} V b_1 \omega | L_2' V b_2 \omega \rangle \\
		& = \lim_n \lim_m \langle L_1' a_n^{\dagger} V (b_1)_m \omega | L_2' V b_2 \omega \rangle \\
		& = \lim_n \lim_m \langle L_1' V a_n^{\dagger} (b_1)_m \omega | L_2' V b_2 \omega \rangle \\
		& = \lim_n \langle L_1' V a_n^{\dagger} b_1 \omega | L_2' V b_2 \omega \rangle \\
		& = \langle L_1' V a^{\dagger} b_1 \omega | L_2' V b_2 \omega \rangle \\
		& = \langle L_1' V a^{*} b_1 \omega | L_2' V b_2 \omega \rangle \\
		& = \langle a^{*} L_1' V b_1 \omega | L_2' V b_2 \omega \rangle.
	\end{split}
\end{align}
So we have constructed a $*$-representation of $\A_0$ on $\H_{\phi}$ in terms of a collection of operators that act on a common domain $\A_{\phi}' V \A_0|\omega\rangle,$ and that satisfy (in the language of appendix \ref{app:unbounded-operators}) the extension formula $a^* \subseteq a^{\dagger}.$
Consequently, any element $a \in \A_0$ has a natural closed extension acting on a subspace of $\H_{\phi}$ including $\A_{\phi}' V \A_0|\omega\rangle.$
By abuse of notation, we will call this closed operator $a$ as well.

In terms of these closed operators, we have
\begin{align}
	\begin{split}
		\hat{\omega}(a \otimes b)
		& = \langle a^* \omega | V^{\dagger} P_{\omega} J_{\phi} P_{\omega} V | b^*\omega\rangle \\
		& = 2 \langle P_{\omega} a^* P_{\omega} \phi | J_{\phi} | P_{\omega} b^* P_{\omega} \phi \rangle.
	\end{split}
\end{align}
Since the operators $a^*$ and $b^*$ are affiliated with $\A_{\phi}$ by construction, similar manipulations to those in the bounded case give the simplified expression
\begin{align} \label{eq:unbounded-omegahat}
	\begin{split}
		\hat{\omega}(a \otimes b)
		& = 2 \langle a^* P_{\omega} \phi | J_{\phi} | b^* P_{\omega} \phi \rangle.
	\end{split}
\end{align}
\subsection{Positivity}
\label{sec:positivity}

Now we have the tools we need to prove positivity of $\hat{\omega}$ in the case that $\omega$ is not separating.
As in section \ref{sec:separating-positivity}, the inequality we wish to establish is
\begin{align}
	\begin{split}
		\sum_{j, k} \hat{\omega}((a_j^* \otimes b_j^*) (a_k \otimes b_k))
			\geq 0
	\end{split}
\end{align}
for any finite sum.

Suppose first that the representation of $\A_0$ on $\H_{\omega}$ is bounded.
Using the techniques of section \ref{sec:separating-extension}, we can embed $\H_{\omega}$ into a GNS space $\H_{\phi}$ with a cyclic and separating vector $|\phi\rangle.$
Using equation \eqref{eq:separating-ext-eq}, we write
\begin{align}
	\begin{split}
		\sum_{j, k} \hat{\omega}((a_j^* \otimes b_j^*) (a_k \otimes b_k))
		& = 2 \sum_{j, k}\hat{\phi}(P_{\omega} a_j^* a_k \otimes P_{\omega} b_k b_j^*) \\
		& = 2 \sum_{j, k} \langle a_k^* a_j P_{\omega} \phi|J_{\phi} | b_j b_k^* P_{\omega} \phi\rangle.
	\end{split}
\end{align}
Since $|\phi\rangle$ is separating, $J_{\phi}$ maps $\A_{\phi}$ to $\A_{\phi}'$ by conjugation, and it is easy to manipulate this expression to obtain
\begin{align}
	\begin{split}
		\sum_{j, k} \hat{\omega}((a_j^* \otimes b_j^*) (a_k \otimes b_k))
		& = 2 \sum_{j, k} \langle a_j J_{\phi} b_j^* P_{\omega} \phi| a_k J_{\phi} b_k^* P_{\omega} \phi\rangle,
	\end{split}
\end{align}
where we have also used the general separating identity $J_{\phi} = J_{\phi}^{\dagger}$, together with the identity $J_{\phi} P_{\omega} |\phi\rangle = P_{\omega}|\phi\rangle$ from section \ref{sec:separating-extension}.
Since this is the norm-squared of a state, it is nonnegative.

In the case where the representation of $\A_0$ is unbounded, we use equation \eqref{eq:unbounded-omegahat} to write
\begin{align}
\begin{split}
	\sum_{j, k} \hat{\omega}((a_j^* \otimes b_j^*) (a_k \otimes b_k))
	& = 2 \sum_{j, k} \langle a_k^* a_j P_{\omega} \phi | J_{\phi} | b_j b_k^* P_{\omega} \phi \rangle.
\end{split}
\end{align}
As explained in section \ref{sec:separating-extension}, the operators $a_j$ and $b_j$ live in a $*$-representation of $\A_0$ on $\H_{\phi},$ and are affiliated to $\A_{\phi}.$
An identical argument to that given in appendix \ref{app:commutant-domain} then tells us that for any $a, b \in \A_0,$ the vector $J_{\phi} b P_{\omega}|\phi\rangle$ is in the domain of $a$ and satisfies the equation
\begin{equation}
	a J_{\phi} b P_{\omega} |\phi\rangle
		= J_{\phi} b J_{\phi} a P_{\omega} |\phi\rangle.
\end{equation}
Applying this identity twice, we get
\begin{align} \label{eq:omegahat-unbounded-isometry}
	\begin{split}
		\sum_{j, k} \hat{\omega}((a_j^* \otimes b_j^*) (a_k \otimes b_k))
		& = 2 \sum_{j, k} \langle a_j P_{\omega} \phi | J_{\phi} b_j b_k^* J_{\phi} | a_k P_{\omega} \phi \rangle \\
		& = 2 \sum_{j, k} \langle J_{\phi} b_j^* J_{\phi} a_j P_{\omega} \phi | J_{\phi} b_k^* J_{\phi} a_k P_{\omega} \phi \rangle \\
		& = 2 \sum_{j, k} \langle a_j J_{\phi} b_j^* P_{\omega} \phi | a_k J_{\phi} b_k^* P_{\omega} \phi \rangle,
	\end{split}
\end{align}
which is nonnegative.

\subsection{Structure of the GNS space}
\label{sec:GNS-isometry}

In the preceding subsection, we showed that the functional $\hat{\omega}$ defined in equation \eqref{eq:intro-main-formula} is positive on $\A_0 \otimes \A_0^{\text{op}}.$
This means that it can be used to produce a GNS space $\H_{\hat{\omega}}$ carrying a representation of $\A_0 \otimes \A_0^{\text{op}}.$
In this section we study the structure of this representation.
In particular, we consider the following maps between GNS spaces associated with $\omega, \hat{\omega}, \phi,$ and $\hat{\phi},$ where $\phi$ was defined in section \ref{sec:separating-extension}.
Note that $L$ is an element of $\A_{\omega}$ and $a$ is an element of $\A_0$.
\begin{align}
	\H_{\omega} \hookrightarrow  \H_{\phi} && L |\omega\rangle \mapsto \sqrt{2} L P_{\omega} |\phi\rangle, \label{eq:separating-inclusion} \\
	\H_{\omega} \hookrightarrow \H_{\hat{\omega}} && a|\omega\rangle \mapsto (a \otimes 1) |\hat{\omega}\rangle, \label{eq:omega-to-omegahat}\\
	\H_{\hat{\phi}} \cong \H_{\phi} && (L_1 \otimes L_2) |\hat{\phi}\rangle \mapsto L_1 J_{\phi} L_2^{\dagger} |\phi\rangle, \label{eq:phihat-to-phi} \\
	\H_{\hat{\omega}} \cong \H_{\phi} && (a \otimes b)|\hat{\omega}\rangle \mapsto \sqrt{2} a J_{\phi} b^* P_{\omega} |\phi\rangle.\label{eq:two-hats}
\end{align}

The first map was shown to be an isometry in section \ref{sec:separating-extension}.
In the case that $\omega$ is not separating, it is not unitary, because it maps --- thanks to the arguments given in section \ref{sec:separating-extension} --- onto the support of the projector $P_{\phi}' = J_{\phi} P_{\omega} J_{\phi}.$

The second map is clearly isometric from the definition of $\hat{\omega}.$
Its non-unitarity will follow from non-unitarity of the first map, once we establish unitarity of the fourth map.

Unitarity of the third map holds automatically because $\phi$ was separating, thanks to the arguments of section \ref{sec:GNS-isomorphism}.

The fourth map is an isometry thanks to equation \eqref{eq:omegahat-unbounded-isometry}, but we will have to do a little extra work to show that it is unitary.
To do this, it suffices to show that it has a dense image.
The density of the right-hand side of the fourth map in $\H_{\phi}$ is equivalent to density of the set $\A_0 J_{\phi} V \A_0 |\omega\rangle,$ where $V$ is the map in equation \eqref{eq:separating-inclusion}.
Density of $\A_0 J_{\phi} V \A_{\omega} |\omega\rangle$ follows from footnote \ref{foot:density}, and the general conclusion follows because elements of $\A_0 |\omega\rangle$ can be approximated arbitrarily well by elements of $\A_{\omega} |\omega\rangle.$\footnote{Even in the case where the representation of $\A_0$ is unbounded, this approximation can be performed so that the limit converges; using the terminology of section \ref{sec:separating-extension}, we write
\begin{equation}
	a J_{\phi} V b |\omega\rangle 
		= \sqrt{2} J_{\phi} b J_{\phi} a P_{\omega} |\omega\rangle 
		= \sqrt{2} \lim_n J_{\phi} b_n J_{\phi} a P_{\omega} |\omega\rangle
		= \lim_n a J_{\phi} V b_n |\omega\rangle.
\end{equation}}

\subsection{Purity in the factorial case}
\label{sec:non-separating-purity}

We now present a generalization of the arguments given in section \ref{sec:separating-purity}, to establish that $\hat{\omega}$ is pure (in the weak sense) if $\omega$ is factorial.

\begin{lemma}
	On $\H_{\hat{\omega}}$, we will call $(\A_0 \otimes 1)_{\hat{\omega}}, (1 \otimes \A_0^{\text{op}})_{\hat{\omega}},$ and $(\A_0 \otimes \A_0^{\text{op}})_{\hat{\omega}}$ the minimal von Neumann algebras generated by the $*$-algebras in parentheses.
	Under the unitary equivalence of equation \eqref{eq:two-hats}, these are mapped respectively to $\A_{\phi}, \A_{\phi}',$ and $\A_{\phi} \vee \A_{\phi}',$ where $\vee$ is the algebraic union.
\end{lemma}
\begin{proof}
	In the bounded case the proof is very straightforward.
	In this setting the elements of $\A_0$ act in a bounded way on both $\H_{\hat{\omega}}$ and $\H_{\phi},$ and it is clear that the unitary from equation \eqref{eq:two-hats} maps $(\A_0 \otimes 1)$ to $\A_0.$
	So the correspondence between $(\A_0 \otimes 1)_{\hat{\omega}}$ and $\A_{\phi}$ follows by taking limits.
	To see the correspondence between $(1 \otimes \A_0^{\text{op}})$ and $\A_{\phi}'$, we denote the unitary in equation \eqref{eq:two-hats} by $U$, and compute
	\begin{align}
		\begin{split}
		U (1 \otimes a) U^{\dagger} \left[ \sqrt{2} b J_{\phi} c^* P_{\omega} |\phi\rangle \right]
			& = U (1 \otimes a) (b \otimes c) |\hat{\omega}\rangle \\
			& = U (b \otimes ca) |\hat{\omega}\rangle  \\
			& = \sqrt{2} b J_{\phi} a^* c^* P_{\omega} |\phi\rangle \\
			& = J_{\phi} a^* J_{\phi} \left[ \sqrt{2} b J_{\phi} c^* P_{\omega} |\phi\rangle \right].
		\end{split}
	\end{align}
	This gives
	\begin{equation}
		U (1 \otimes a) U^{\dagger}
			= J_{\phi} a^* J_{\phi},
	\end{equation}
	and $U$ conjugates a dense subset of $(1 \otimes \A_0^{\text{op}})_{\hat{\omega}}$ to a dense subset of $\A_{\phi}'.$
	This establishes the second correspondence.
	The final correspondence holds so long as one has
	\begin{equation} \label{eq:vee-lemma}
		(\A_0 \otimes \A_0^{\text{op}})_{\hat{\omega}}
			= (\A_0 \otimes 1)_{\hat{\omega}} \vee (1 \otimes \A_0^{\text{op}})_{\hat{\omega}}.
	\end{equation}
	The inclusion $\supseteq$ is obvious, since $\A_0 \otimes \A_0^{\text{op}}$ contains dense subspaces of both $(\A_0 \otimes 1)_{\hat{\omega}}$ and $(1 \otimes \A_0^{\text{op}})_{\hat{\omega}}.$
	But the right-hand side clearly includes $\A_0 \otimes \A_0^{\text{op}},$ which is a dense subspace of the left-hand side, which establishes the inclusion $\subseteq.$
	
	In the unbounded case, one must be a little more careful, but the argument is basically the one given in section \ref{sec:separating-purity}.
	Under unitary conjugation by $U$, $\A_0 \otimes 1$ is mapped to a $*$-algebra of unbounded operators on $\H_{\phi}$ with common domain equal to the image of $U$.
	Concretely, for any $a, b, c \in \A_0,$ we have
	\begin{equation}
		(U (a \otimes 1) U^{\dagger}) \left[ \sqrt{2} b J_{\phi} c^* P_{\omega} |\phi\rangle \right]
			= a b J_{\phi} c^* P_{\omega} |\phi\rangle.
	\end{equation}
	From this we see that $U$ maps $(a \otimes 1)$ on $\H_{\hat{\omega}}$ to the operator $a$ that acts on $\H_{\phi}$ according to the rules of section \ref{sec:separating-extension}.
	But this operator was constructed explicitly so that the bounded functions of its polar decomposition on $\H_{\phi}$ would match those on $\H_{\omega}.$
	So the von Neumann algebra generated by $U (\A_0 \otimes 1) U^{\dagger}$ on $\H_{\phi}$ is the same as the algebra generated by $\A_{\omega}$ on $\H_{\phi}$; as explained in section \ref{sec:separating-extension}, this is all of $\A_{\phi}.$
	
	With this reasoning in mind, a nearly identical argument to the one given in section \ref{sec:separating-purity} establishes
	\begin{equation}
		U (1 \otimes \A_0^{\text{op}})_{\hat{\omega}} U^{\dagger}
			= \A_{\phi}'.
	\end{equation}
	So it remains only to establish equation \eqref{eq:vee-lemma}, or equivalently its commutant form
	\begin{equation}
		(\A_0 \otimes \A_0^{\text{op}})_{\hat{\omega}}'
		= (\A_0 \otimes 1)_{\hat{\omega}}' \cap (1 \otimes \A_0^{\text{op}})_{\hat{\omega}}'.
	\end{equation}
	But the argument given for this in section \ref{sec:separating-purity} relied in no way on the fact that $\omega$ was separating; so we can simply repeat that argument, and we are done.
\end{proof}

\begin{corollary}
	In the factorial case, we have
	\begin{equation}
		(\A_0 \otimes \A_0^{\text{op}})_{\hat{\omega}}'
			= \comps.
	\end{equation}
\end{corollary}
\begin{proof}
	By the lemma, it suffices to show that the commutant of $\A_{\phi} \vee \A_{\phi}'$ on $\H_{\phi}$ is trivial.
	But the commutant of $\A_{\phi} \vee \A_{\phi}'$ is the center $\A_{\phi} \cap \A_{\phi}'.$
	Since $\A_{\phi}$ is isomorphic to $\A_{\omega}$, this center is trivial when $\omega$ is factorial.
\end{proof}

\section{Comparison to the natural cone}
\label{sec:comparison}

\subsection{General considerations}

Thus far we have defined and studied, given a state $\omega : \A_0 \to \comps,$ a state $\hat{\omega} : \A_0 \otimes \A_0^{\text{op}} \to \comps$ that we have called the canonical purification.
We call this ``canonical'' because it mimics the time-reflection-sewing character of the finite-dimensional canonical purification reviewed in section \ref{sec:finite-dimension}.
However, there are two other kinds of purifications that have been referred to as ``canonical'' in the literature: (i) the GNS purification, and (ii) the purification into the natural cone.
We already explained, in section \ref{sec:time-reversal-GNS} for the bounded case and in sections \ref{sec:GNS-isomorphism} and \ref{sec:GNS-isometry} for the unbounded case, the connection between GNS and our canonical purification.
Here we address the connection to the natural cone.

The natural cone is reviewed in detail in appendix \ref{app:natural-cone}.
The idea is that given a Hilbert space $\H$ with a von Neumann algebra $\A_0$ and a cyclic-separating state $|\omega\rangle,$ one constructs a set $\mathfrak{P}_{1/4}$ of states such that all cyclic-separating elements of $\mathfrak{P}_{1/4}$ have modular conjugations matching the modular conjugation of $|\omega\rangle.$
Every appropriately continuous functional $\psi$ of $\A_0$ has a unique representative within $\mathfrak{P}_{1/4},$ called its ``natural purification with respect to $|\omega\rangle.$''
Calling this ``canonical'' can be misleading, as if one changes the state $|\omega\rangle,$ then one generically changes the natural purification of $\psi,$ so the natural purification of $\psi$ is defined in reference to a particular cyclic-separating state $|\omega\rangle,$ or at least in reference to the family of states $\mathfrak{P}_{1/4}.$

The connection between the natural purification and reflection-sewing comes from the fact that the vector representative $|\psi_{1/4}\rangle$ is fixed by the modular conjugation $J_{\omega}.$
If $|\omega\rangle$ is the Minkowski vacuum and $\A_0$ is a Rindler wedge, then by the Bisognano-Wichmann theorem \cite{Bisognano-Wichmann}, the modular conjugation $J_{\omega}$ is the CRT reflection map that maps $\A_0$ to the complementary wedge.
Since in this case $|\psi_{1/4}\rangle$ is symmetric under CRT-reflection, one can roughly think of it as two copies of $\psi$ that are sewed together after CRT conjugation.
This interpretation, however, is fairly special to the Minkowski vacuum in the Rindler wedge.
Moreover, natural purifications are only defined for ultraweakly continuous states on von Neumann algebras, while the canonical purification presented in this paper is defined more generally.
Nevertheless, we will now demonstrate that there is an interesting relationship between canonical purifications and natural purifications in settings where both can be defined.

\subsection{Matching canonical to natural purifications}

Suppose we have two states $\omega_1, \omega_2$ defined on the same abstract $*$-algebra $\A_0.$
We will assume they are separating, so that the states $|\hat{\omega}_1\rangle$ and $|\hat{\omega}_2\rangle$ are cyclic and separating, within their respective GNS spaces, for the von Neumann algebras generated by $(\A_0 \otimes 1)$.
Moreover, we suppose there is a unitary equivalence $U$ between the GNS spaces $\H_{\hat{\omega}_1}$ and $\H_{\hat{\omega}_2}$ that preserves the action of $\A_0 \otimes \A_0^{\text{op}}.$
This is a unitary map
\begin{equation} \label{eq:equivalence-U-domains}
	U : \H_{\hat{\omega}_1} \to \H_{\hat{\omega}_2}
\end{equation}
that satisfies
\begin{equation} \label{eq:equivalence-U-action}
	U(a \otimes b) U^{\dagger} = a \otimes b.
\end{equation}
It does not map the state $|\hat{\omega}_1\rangle \in \H_{\hat{\omega}_1}$ to the state $|\hat{\omega}_2\rangle \in \H_{\hat{\omega}_2}$; instead, it maps $|\hat{\omega}_1\rangle$ to some excited state $U |\hat{\omega}_1\rangle.$
We claim that $U |\hat{\omega}_1\rangle$ can always be brought into the natural cone of $|\hat{\omega}_2\rangle$ by acting with a unitary $W$ in the center of $(\A_0 \otimes 1)_{\hat{\omega}_2}.$
We will do this in two steps.
For simplicity of exposition, we start with the factorial case, where $W$ will just be a phase, then give the general proof in the next subsection.

A theorem due to Araki in \cite{Araki:natural-1}, reviewed in appendix \ref{app:natural-cone}, tells us that for $U |\hat{\omega}_1\rangle$ to be in the natural cone of $|\hat{\omega}_2\rangle,$ we must have equality of the modular conjugations:
\begin{equation} \label{eq:sec-6-conj-equality}
	J_{U \hat{\omega}_1} = J_{\hat{\omega}_2}
\end{equation}
together with the inequality
\begin{equation}
	\langle U \hat{\omega}_1 | Z | \hat{\omega}_2\rangle \geq 0
\end{equation}
for every positive $Z$ in the center.
In the factorial case, this second condition is simply
\begin{equation} \label{eq:factorial-inequality}
	\langle U \hat{\omega}_1 | \hat{\omega}_2\rangle \geq 0.
\end{equation}

To start, we observe that because $U$ is unitary, it conjugates $(\A_0 \otimes 1)_{\hat{\omega}_1}$ to $(\A_0 \otimes 1)_{\hat{\omega}_2}.$
Since we assumed that $|\hat{\omega}_1\rangle$ was cyclic and separating for the first of these algebras, the state $U |\hat{\omega}_1\rangle$ is cyclic and separating for the second algebra.
From the defining formula for the Tomita operator it is easy to check
\begin{equation}
	S_{U \hat{\omega}_1} = U S_{\hat{\omega}_1} U^{\dagger},
\end{equation}
and by studying the polar decomposition one finds the modular conjugation formula
\begin{equation}
	J_{U \hat{\omega}_1}
	= U J_{\hat{\omega}_1} U^{\dagger}.
\end{equation}
From this one easily computes
\begin{equation}
	J_{U \hat{\omega}_1} (a \otimes b) U |\hat{\omega}_1\rangle
	= (b^* \otimes a^*) U |\hat{\omega}_1\rangle.
\end{equation}
So in particular, one has
\begin{equation}
	J_{U \hat{\omega}_1} (a \otimes b) J_{U \hat{\omega}_1}
	= J_{\hat{\omega}_2} (a \otimes b) J_{\hat{\omega}_2},
\end{equation}
and by taking limits one sees that for any $L \in (\A_0 \otimes \A_0^{\text{op}})_{\hat{\omega}_2}$ one has
\begin{equation}
	J_{U \hat{\omega}_1} L J_{U \hat{\omega}_1}
	= J_{\hat{\omega}_2} L J_{\hat{\omega}_2}.
\end{equation}
From this we deduce that the unitary operator $J_{\hat{\omega}_2} J_{U \hat{\omega}_1}$ is in the commutant of $(\A_0 \otimes \A_0^{\text{op}})_{\hat{\omega}_2},$ which --- thanks to the structure elaborated in section \ref{sec:separating-purity} --- is the center of $(\A_0 \otimes 1)_{\hat{\omega}_2}.$
In the factorial setting, this means that there is a phase with
\begin{equation}
	J_{U \hat{\omega}_1} = e^{i \theta} J_{\hat{\omega}_2}.
\end{equation}
If we multiply $U$ by $e^{-i \theta/2},$ we find
\begin{equation}
	J_{e^{-i \theta/2} U \hat{\omega}_1}
	= e^{-i \theta/2} J_{U \hat{\omega}_1}e^{i \theta/2}
	= e^{- i \theta} J_{U \hat{\omega}}
	= J_{\hat{\omega}_2}.
\end{equation}
We will now absorb this phase into $U$, which guarantees equality in \eqref{eq:sec-6-conj-equality}.

Next, we must modify $U$ so that equation \eqref{eq:factorial-inequality} holds.
Crucially, however, we must make this modification without changing the modular conjugation $J_{U \hat{\omega}_1}.$
The overlap $\langle U \hat{\omega}_1 | \hat{\omega}_2\rangle$ is a complex number, so we could clearly multiply $U$ by the phase of this complex number to guarantee inequality \eqref{eq:factorial-inequality}, but this will change the modular conjugation unless the phase is real.
Luckily, the phase \textit{is} real; this follows from the manipulation
\begin{equation}
	\langle U \hat{\omega}_1 | \hat{\omega}_2 \rangle
	= \langle J_{\hat{\omega}_2} U \hat{\omega}_1 | \hat{\omega}_2 \rangle
	= \langle J_{\hat{\omega}_2} \hat{\omega}_2 | U \hat{\omega}_1\rangle
	= \langle \hat{\omega}_2 | U \hat{\omega}_1\rangle.
\end{equation}
So this is a positive or negative number; if it is positive we make no change, and if it is negative we change the sign of $U$.
The resulting state $U |\hat{\omega}_1\rangle$ satisfies equation \eqref{eq:sec-6-conj-equality} and inequality \eqref{eq:factorial-inequality}, so it is in the natural cone of $|\hat{\omega}_2\rangle,$ as desired.

\subsection{Canonical/natural equivalence in the non-factorial case}

Now we generalize the proof of the preceding subsection to the non-factorial setting.
As in the factorial case, we want to modify $U$ by a unitary in the center to obtain the equality of modular conjugations
\begin{equation}
	J_{U \hat{\omega}_1} = J_{\hat{\omega}_2}
\end{equation}
and the inequality
\begin{equation} \label{eq:sec-6-2-inequality}
	\langle U \hat{\omega}_1 | Z |\hat{\omega}_2\rangle \geq 0
\end{equation}
for positive $Z$ in the center.
As in the factorial case, we will accomplish this in two steps.
The general setting is more complicated, however, and both steps will require applying the spectral theorem in its general, infinite-dimensional form.
This can be reviewed in \cite[chapters 12 and 13]{Rudin:functional-book}.

First, we follow the initial steps of the factorial proof to conclude that
the operator $J_{\hat{\omega}_2} J_{U \hat{\omega}_1}$ is in the center.
Because this operator is unitary, its spectrum is entirely contained in the complex unit circle.
So if we define the function $f$ on this circle by
\begin{equation}
	f(e^{i \theta}) = e^{i \theta/2}, \qquad \theta \in [0, 2\pi),
\end{equation}
then the operator
\begin{equation}
	W = f(J_{\hat{\omega}_2} J_{U \hat{\omega}_1})
\end{equation}
is a unitary operator in the center that squares to $J_{\hat{\omega}_2} J_{U \hat{\omega}_1}$.
One has
\begin{equation}
	J_{\hat{\omega}_2} (J_{\hat{\omega}_2} J_{U \hat{\omega}_1}) J_{\hat{\omega}_2}
		= J_{U \hat{\omega}_1} J_{\hat{\omega}_2},
\end{equation}
and approximating $f$ by complex polynomials one finds\footnote{The function $f$ is continuous except at $\theta=2\pi,$ so it can be approximated by a uniformly bounded sequence of continuous functions that differ from $f$ only in the interval $\theta \in [2 \pi -1/n, 2 \pi)$. By the Stone-Weierstrass theorem, these continuous functions can in turn be approximated uniformly by polynomials.
Putting this together, one can find a uniformly bounded sequence of polynomials $p_n(z, \bar{z})$ such that each $p_n$ is within $1/n$ distance of $f$ on the interval $\theta \in [0, 2 \pi-(1/n)].$
For any state $|\psi\rangle$, the spectral theorem gives a complex measure $\mu_{\psi, \psi}$ on the circle with
\begin{equation}
	\lVert p_n(J_{U \hat{\omega}_1} J_{\hat{\omega}_2}) |\psi\rangle - W |\psi\rangle \rVert^2
		= \int \mu_{\psi, \psi}(\theta) |p_n(e^{i \theta}) - f(e^{i \theta})|^2,
\end{equation}
and similarly for $p_n^*, W^{\dagger},$ and $f^*.$
The assumptions we have made so far then give the strong-limit formulas
$$W = f(J_{\hat{\omega}_2} J_{U \hat{\omega}_1}) = \lim_n p_n(J_{\hat{\omega}_2} J_{U \hat{\omega}_1}, J_{U \hat{\omega}_1} J_{\hat{\omega}_2})$$
and
$$W^{\dagger} = f^*(J_{\hat{\omega}_2} J_{U \hat{\omega}_1}) = \lim_n p_n^*(J_{\hat{\omega}_2} J_{U \hat{\omega}_1}, J_{U \hat{\omega}_1} J_{\hat{\omega}_2}).$$
Conjugating by $J_{\hat{\omega}_2}$ complex-conjugates the polynomial $p_n,$ which proves the claim.}
\begin{equation} \label{eq:J-W-conjugation}
	 J_{\hat{\omega}_2} W J_{\hat{\omega}_2} 
	 = W^{\dagger}.
\end{equation}

We then have
\begin{equation}
	J_{W U \hat{\omega}_1}
		= W J_{U \hat{\omega}_1} W^{\dagger}
		= W J_{\hat{\omega}_2} W^2 W^{\dagger}
		= W J_{\hat{\omega}_2} W.
\end{equation}
One then uses equation \eqref{eq:J-W-conjugation} to compute
\begin{equation}
	J_{W U \hat{\omega}_1}
		= J_{\hat{\omega}_2} (J_{\hat{\omega}_2} W J_{\hat{\omega}_2}) W = J_{\hat{\omega}_2}.
\end{equation}
We will now absorb $W$ into a redefinition of $U$; this gives us a unitary map $U$ satisfying equations \eqref{eq:equivalence-U-domains} and \eqref{eq:equivalence-U-action}, and with $J_{U \hat{\omega}_1} = J_{\hat{\omega}_2}.$
We now only need to check inequality \eqref{eq:sec-6-2-inequality} on the center.
	
The center, being an abelian von Neumann algebra on a separable Hilbert space, can be generated by bounded functions of a single Hermitian operator $T$ \cite{v1930algebra}.
So what we really want is for any bounded, positive function $p$ on the real line, the inequality
\begin{equation}
	\langle U \hat{\omega}_1 | p(T) | \hat{\omega}_2 \rangle \geq 0.
\end{equation}
Note that the left-hand side is automatically real, since we have already shown that $U |\hat{\omega}_1\rangle$ and $|\hat{\omega}_2\rangle$ have the same modular conjugation; this gives 
\begin{equation}
	\langle U \hat{\omega}_1 | p(T) \hat{\omega}_2 \rangle
		= \langle J_{\hat{\omega}_2} U \hat{\omega}_1 | p(T) \hat{\omega}_2 \rangle
		= \langle J_{\hat{\omega}_2} p(T) \hat{\omega}_2 | U \hat{\omega}_1\rangle
		= \langle p(T) \hat{\omega}_2 | U \hat{\omega}_1\rangle,
\end{equation}
where in the last step we have used that $p(T)$ is in the center and applied the centralizer theorem (cf. appendix \ref{app:centralizer}).

The spectral theorem gives a complex measure $\nu$ on the real line satisfying
\begin{equation}
	\langle U \hat{\omega}_1 | p(T) | \hat{\omega}_2 \rangle
		= \int d\nu(t)\, p(t).
\end{equation}
Because all of the overlaps are real, we actually know that $\nu$ is a real measure. 
Every such measure can be split into positive and negative parts with disjoint support (see e.g. \cite[theorem 6.14]{Rudin:measure-book}).
If we define $\zeta$ to be a function that is $+1$ where $\nu$ is positive and $-1$ where $\nu$ is negative, then we can write $\nu$ in terms of a positive measure $|\nu|$ as
\begin{equation}
	\int d\nu(t)\, p(t)
	= \int d|\nu|(t)\, \zeta(t) p(t).
\end{equation}
If we define $W = \zeta(T),$ then this is a unitary in the center (in fact it is also self-adjoint), and we have
\begin{equation}
	\langle W U \hat{\omega}_1 | p(T) |\hat{\omega}_2\rangle
		= \langle U \hat{\omega}_1 | \zeta^*(T) p(T) |\hat{\omega}_2\rangle 
		= \int d|\nu|(t) \zeta(t) \zeta^*(t) p(t) \geq 0.
\end{equation}
Moreover, this does not change the modular conjugation, since we have
\begin{equation}
	J_{W U \hat{\omega}_1}
		= \zeta(T) J_{U \hat{\omega}_1} \zeta(T)
		= \zeta(T) (J_{U \hat{\omega}_1} \zeta(T) J_{U \hat{\omega}_1}) J_{U \hat{\omega}_1}
		= \zeta(T) (J_{\hat{\omega}_2} \zeta(T) J_{\hat{\omega}_2}) J_{U \hat{\omega}_1}
\end{equation}
and the desired conclusion follows if $\zeta(T)$ is sent to itself by modular conjugation.
But this follows from the fact that $\zeta(T)$ is in the center.
More concretely, for any $L$ in the von Neumann algebra generated by $(\A_0 \otimes 1),$ we have
\begin{align}
	\begin{split}
	(J_{\hat{\omega}_2} \zeta(T) J_{\hat{\omega}_2}) L | \hat{\omega}_2\rangle
		& = J_{\hat{\omega}_2} \zeta(T) (J_{\hat{\omega}_2} L J_{\hat{\omega}_2}) | \hat{\omega}_2\rangle \\
		& = J_{\hat{\omega}_2} (J_{\hat{\omega}_2} L J_{\hat{\omega}_2}) \zeta(T) | \hat{\omega}_2\rangle \\
		& = L J_{\hat{\omega}_2} \zeta(T) | \hat{\omega}_2\rangle \\
		& = L \zeta(T) |\hat{\omega}_2\rangle \\
		& = \zeta(T) L |\hat{\omega}_2\rangle.
	\end{split}
\end{align}
In an intermediate step we have once again used the centralizer theorem to conclude that the state $\zeta(T) |\hat{\omega}_2\rangle$ is fixed by $J_{\hat{\omega}_2}.$

We conclude that $W U |\hat{\omega}_1\rangle$ is in the natural cone of $|\hat{\omega}_2\rangle.$


\acknowledgments{
I thank Jackie Caminiti and Federico Capeccia for careful checking of certain arguments in this paper and for collaboration on related work.
The presentation of this paper was significantly improved by conversations with Jonah Kudler-Flam and Geoff Penington at the Perimeter Institute conference ``QIQG 2025.''}

\appendix 

\section{von Neumann algebras from $*$-algebras}
\label{app:star-algebras}

A general quantum theory does not admit a preferred, intrinsic description in terms of von Neumann algebras.
This is because a single theory may admit many different Hilbert space sectors, in which the local 
degrees of freedom will be completed into distinct von Neumann algebras.
This is a problem already at the level of free field theory in curved spacetime (see for example \cite{Hollands:review}), and has recently been proposed as an important aspect of quantum gravity as experienced by an observer carrying internal degrees of freedom \cite{Witten:background}.

Instead, a quantum theory should be described by a collection of $*$-algebras, potentially with some additional data like an operator product expansion \cite{Hollands:axioms}.
The $*$-algebra is a collection of formal symbols that describe the local degrees of freedom of the theory --- e.g., the local fields of a quantum field theory --- that can be combined by summation or multiplication, and for which one can take adjoints.
Given a set of correlation functions for this algebra, one can perform the GNS construction (reviewed below) to obtain a Hilbert space on which the elements of the $*$-algebra act as operators.
However, because these operators are unbounded, they do not live in a von Neumann algebra, and one cannot obviously ``complete'' the set of unbounded operators in a region to form a von Neumann algebra of local observables.

In this appendix, we give a general construction of von Neumann algebras generated by GNS representations of a $*$-algebra, and prove cyclicity theorems and purity criteria that will be used in the main text.
Our approach is slightly different from existing approaches, because we only try to create a single von Neumann algebra associated with the full GNS representation, not a full net associated to a family of subalgebras.
For existing approaches, which produce a more robust algebraic structure but require additional restrictions on the representation, see \cite{Driessler:unbounded-to-vN, Buchholz:unbounded-to-vN}.\footnote{For experts, we note that the main difference between this paper and \cite{Driessler:unbounded-to-vN, Buchholz:unbounded-to-vN}  is that we define our von Neumann algebra using the intersection of the strong commutant and its adjoint set. This always defines a von Neumann algebra, though it will differ in general from the definition using the weak commutant. However we note that in \cite{Driessler:unbounded-to-vN, Buchholz:unbounded-to-vN} one generally seeks conditions under which in which the weak and strong commutants are equal.
For general technical properties of these commutants, see \cite{Powers:self-adjoint, Borchers:weak-commutant, Schmudgen:commutants}}

\subsection{Key points}

\begin{itemize}
	\item We think of a quantum theory as an abstract $*$-algebra $\A_0$, potentially with additional data, on which one may define states as positive linear functionals $\omega: \A_0 \to \comps.$
	Every such functional gives rise to a GNS representation $\H_{\omega}$ with GNS domain $\{a |\omega\rangle\},$ on which $\A_0$ acts by left multiplication.
	These operators are generally unbounded.
	\item On a GNS Hilbert space $\H_{\omega},$ each $a \in \A_0$ has a minimal ``closure'' $\bar{a}$ admitting a polar decomposition
	\begin{equation}
		\bar{a}
			= V_{\bar{a}} |\bar{a}|.
	\end{equation}
	By abuse of notation, we write $a = \bar{a}.$
	\item A closed operator $T$ is said to be affiliated to a von Neumann algebra $\M$ if $\M$ contains $V_T$ and all bounded functions of $|T|.$
	Equivalently, if every operator in the commutant $\M'$ preserves the domain of $T$ and commutes with $T$ on that domain.
	In this case one says that the operators in $\M'$ ``commute strongly'' with $T$, and writes
	\begin{equation}
		L' T \subseteq L a' \qquad \forall L' \in \M'.
	\end{equation}
	\item We define $\A_{\omega}$ to be the smallest von Neumann algebra on $\H_{\omega}$ such that $a$ is affiliated to $\A_{\omega}$  for every $a \in \A_0.$
	This is the von Neumann algebra generated by all partial isometries $V_{a}$ and bounded functions of self-adjoint operators $|a|.$
	Equivalently, one may define $\A_0'$ as the set of all operators that commute strongly with $\A_0$ in the sense of the preceding bullet point, and define
	\begin{equation}
		\A_{\omega} = (\A_0' \cap (\A_0')^{\dagger})'.
	\end{equation}
	The intersection on the right-hand side is already a von Neumann algebra, so we have
	\begin{equation}
		\A_{\omega}' = \A_0' \cap (\A_0')^{\dagger}.
	\end{equation}
	\item The GNS vector $|\omega\rangle$ is cyclic for $\A_{\omega},$ meaning that vectors of the form $L |\omega\rangle$ with $L \in \A_{\omega}$ are dense in $\H_{\omega}.$
	\item In the case that $|\omega\rangle$ is also separating for $\A_{\omega},$ the vector $J_{\omega} b |\omega\rangle$ is in the domain of every $a \in \A_0,$ and satisfies
	\begin{equation}
		a J_{\omega}b|\omega\rangle = J_{\omega} b J_{\omega} a |\omega\rangle.
	\end{equation}
	\item There are two senses in which a state may be said to be pure: we may say (i) it is pure if $\omega$ cannot be written as a nontrivial mixture of other algebraic states; or (ii) it is pure if $\A_{\omega}$ is the full algebra of bounded operators on $\H_{\omega},$ so that the GNS representation contains no ``purifying system.''
	These notions are the same when the GNS representation of $\A_0$ is bounded, but in general they can be different.
	Condition (i) is equivalent to the statement that the weak commutant of $\A_0$ is trivial, meaning that if we have a bounded $O$ with
	\begin{equation}
		\langle \psi | O a |\chi \rangle = \langle a^* \psi | O |\chi \rangle \qquad \forall\, |\psi\rangle, |\chi\rangle \in \DGNS, \qquad \forall \, a \in \A_0,
	\end{equation}
	then $O$ is a multiple of the identity.
	Condition (ii) is more closely linked to the strong commutant discussed in earlier bullet points.
	It is possible for the weak commutant to be nontrivial while the strong commutant is trivial; see for example \cite{Powers:self-adjoint, Borchers:weak-commutant}.
	This means that one can have a state that is pure in sense (ii) but not in sense (i).
	In the main text we show that canonical purifications in the factorial setting are pure in sense (ii).
\end{itemize}

\subsection{Terminology around unbounded operators}
\label{app:unbounded-operators}

Throughout this appendix and in the main paper it is important to deal with unbounded operators.
This subsection provides a brief overview of the necessary terminology and facts; for a more robust review, see \cite[section 2.3]{Sorce:modular}.

On a Hilbert space $\H$, an unbounded operator $T$ is a linear map
\begin{equation}
	T : D_T \to \H,
\end{equation}
where $D_T$ is a subspace of $\H$ called the \textit{domain} of $T$.
We typically assume that $D_T$ is dense in $\H$, in which case $T$ is said to be \textit{densely defined}.
Any densely defined operator has an adjoint $D_{T^{\dagger}},$ which is defined by the standard equation
\begin{equation}
	\langle T^{\dagger} \chi | \psi \rangle = \langle \chi | T \psi\rangle, \quad |\psi\rangle \in D_T.
\end{equation}
In particular, the domain of $D_{T^{\dagger}}$ is defined to be the set of all vectors $|\chi\rangle$ for which there exists a vector $T^{\dagger} |\chi\rangle$ satisfying the above equation.
By the Riesz lemma, this is equivalent to the statement that
\begin{equation}
	|\psi\rangle \mapsto \langle \chi | T \psi\rangle 
\end{equation}
is bounded as a map from $D_T$ to $\mathbb{C}.$

An operator $T$ is said to be \textit{closed} if all converging sequences converge consistently, i.e., if whenever we have a sequence $|\psi_n\rangle \in D_T$ such that $|\psi_n\rangle$ and $T |\psi_n\rangle$ both converge, then we have 
\begin{equation}
	\left(\lim_{n} |\psi_n\rangle\right) \in D_T
\end{equation}
and
\begin{equation}
	\lim_{n} \left(T |\psi_n\rangle\right)
		= T \left( \lim_n |\psi_n\rangle \right).
\end{equation}
An operator is \textit{closable} (or \textit{preclosed}) if it can be made closed by adding limits of sequences to its domain.
In this case one writes $\bar{T}$ for the closure of $T$.

An extension of an operator $T$ is an operator $S$ with a domain $D_S$ containing $D_T$, such that $S$ and $T$ agree on the shared domain.
One writes $T \subseteq S.$
An operator is closable if it admits an extension which is itself closed.
Adjoints are automatically closed, so one can often show that an operator is closable by showing that it has an extension as an adjoint of another operator.

For any closed operator $T$, the operator $T^{\dagger} T$ is self-adjoint with positive spectrum.
By the spectral theorem, one can take the square root to define $|T| = \sqrt{T^{\dagger} T},$ and one can show $D_{|T|} = D_T.$
From this, one can construct for any closed operator $T$ a polar decomposition
\begin{equation}
	T = V_T |T|,
\end{equation}
where $V_T$ is a partial isometry supported on the closure of the image of $|T|.$

Note that because $|T|$ is self-adjoint, one can use the spectral theorem to act on it with arbitrary bounded functions to create a large class of bounded operators ``made out of $T$.''
A closed operator $T$ is said to be \textit{affiliated} to a von Neumann algebra $\M$ if $\M$ contains both (i) the partial isometry $V_T$, and (ii) every bounded function of $|T|.$
Equivalently, $T$ is affiliated to $\M$ if it formally commutes with every element of the commutant $\M'$: if for each $L' \in \M',$ we have
\begin{equation}
	L' : D_{T} \to D_{T}
\end{equation}
and
\begin{equation}
	T L'|_{D_T} = L' T.
\end{equation}
This is equivalent to the extension formula $L' T \subseteq T L'.$
When this holds one sometimes says that $L'$ commutes ``strongly'' with $T$.

\subsection{GNS representations of a $*$-algebra}
\label{app:GNS}

A $*$-algebra $\A_0$ is an abstract algebra that is not necessarily represented on a Hilbert space, but which nevertheless has operations of multiplication, addition, scalar multiplication, and adjoint.
The multiplication and adjoint operations distribute over addition, and the adjoint operation $*$ satisfies
\begin{align}
	(a b)^*
		& = b^* a^*, \quad a, b \in \A, \\
	(\lambda a)^*
		& = \lambda^* a^*, \quad \lambda \in \mathbb{C}.
\end{align}
The $*$-algebra $\A_0$ is assumed to have an identity element, which is just written as $1,$ and which satisfies $1^* = 1.$

An \textit{algebraic state} on a $*$-algebra is a linear map
\begin{equation}
	\omega : \A_0 \to \comps
\end{equation} 
satisfying the Hermiticity condition
\begin{equation}
	\omega(a)^*
		= \omega(a^*)
\end{equation}
and the positivity condition
\begin{equation}
	\omega(a^* a)
		\geq 0.
\end{equation}
We will always take $\omega$ to be normalized with $\omega(1) = 1.$
The map $\omega$ should be thought of as a list of expectation values for ``operators'' in $\A_0,$ or equivalently as a list of correlation functions.
Given any $\omega,$ one can canonically construct a Hilbert space $\H_{\omega}$ containing (i) a state $|\omega\rangle$; and (ii) a representation of the $*$-algebra $\pi_{\omega}(a)$; such that all correlation functions are reproduced via the formula
\begin{equation}
	\omega(a)
		= \langle \omega | \pi_{\omega}(a) |\omega\rangle.
\end{equation}
We will summarize the construction below; for details and proofs, see \cite[chapter 7]{conway2000course}.\footnote{Note that while this source only explicitly covers the case of $C^*$ algebras, the proofs in the $*$-algebraic case are identical.}

To perform the GNS construction, one simply starts with an abstract complex vector space spanned by the formal symbols
\begin{equation}
	\{ |a \omega\rangle\,|\, a \in \A_0 \}.
\end{equation}
One then imposes by hand the inner product
\begin{equation}
	\langle a \omega | b \omega \rangle
		= \omega(a^* b),
\end{equation}
which can be seen by positivity of $\omega$ to be a ``semi-inner product,'' i.e., an inner product that is positive semidefinite but not necessarily positive definite.
Taking the quotient by null states yields a genuine inner product space, and completing with respect to the inner product gives a Hilbert space $\H_{\omega}.$
The equivalence class $[|a \omega\rangle]$ is often just written $|a \omega\rangle,$ and the $*$-algebra $\A_0$ acts on $\H_{\omega}$ by
\begin{equation}
	\pi_{\omega}(a) |b \omega\rangle = |a b \omega\rangle.
\end{equation}
We will henceforth stop writing the representation $\pi_{\omega},$ since we think omitting it causes no confusion, and will freely manipulate expressions like $a |\omega\rangle.$
Note that $a$ will generally be an unbounded operator; it is defined on the domain spanned by vectors of the form $|b \omega\rangle,$ which are dense in $\H_{\omega}$.
This domain is called the \textit{GNS domain} $\DGNS$.

We now must deal with a serious subtlety in GNS representations of $*$-algebras: we do not generally have the identity $a^* = a^{\dagger}.$
The reason for this is that $a^*$ is the GNS representation of the abstract algebra element $a^*,$ while $a^{\dagger}$ is the adjoint of the GNS representation of the abstract algebra element $a.$
In particular, $a^*$ is only defined on the GNS domain, while $a^{\dagger}$ may have a much bigger domain!
Concretely, we have
\begin{equation}
	\langle a^* b \omega | c \omega\rangle 
		= \langle b \omega | a c \omega\rangle,
\end{equation}
from which we see that the domain of $a^{\dagger}$ includes $\DGNS,$ and that $a^{\dagger}$ and $a^*$ agree on $\DGNS$.
Using the discussion of the previous subsection, we may conclude from this the extension formula
\begin{equation}
	a^* \subseteq a^{\dagger},
\end{equation}
which implies that $a^*$ is a closable operator.
However, it is possible to have
\begin{equation} \label{eq:app-proper-extension}
	\bar{a^*} \subsetneq a^{\dagger},
\end{equation}
and there are physical examples where this inclusion is proper (see e.g. \cite{Rabsztyn:deficiency}).

This subtlety makes the unbounded operators generated by physical degrees of freedom \textit{extremely} complicated to deal with in general.
If expression \eqref{eq:app-proper-extension} were an equality, then one could generate a von Neumann algebra on $\H_{\omega}$ by taking all Hermitian elements of $\A_0$ (those with $a = a^*$), turning them into self-adjoint operators on $\H_{\omega}$ by taking closures, and generating a von Neumann algebra by acting on these self-adjoint operators with bounded functions.
Instead, one must be more careful.

\subsection{von Neumann algebras in a GNS representation}
\label{app:vN}

From the previous subsection, we know that in a GNS representation, we have a Hilbert space $\H_{\omega}$ with a dense ``GNS domain'' $\DGNS$ on which a $*$-algebra is represented.
We have $a^* \subseteq a^{\dagger},$ so each GNS operator is closable.
We would like to use this data to construct a von Neumann algebra, so that we may use the tools of modular theory to study physics.
Inspired by the discussion of subsection \ref{app:unbounded-operators}, we will construct $\A_{\omega}$ such that it is the minimal von Neumann algebra for which all elements of $\A_0$ are affiliated.
Concretely, each element of $\A_0$ gives rise to a closable operator $a,$ and its closure $\bar{a}$ has a polar decomposition
\begin{equation}
	\bar{a} = V_{\bar{a}} |\bar{a}|.
\end{equation}
We take $\A_{\omega}$ to be the von Neumann algebra generated by all partial isometries $V_{\bar{a}}$ and all bounded functions of operators $|\bar{a}|$, for all $a \in \A_0$.\footnote{Given a set $M$ of bounded operators, the von Neumann algebra generated by $M$ is the double commutant of the set $M \cup M^{\dagger}$. This is the smallest von Neumann algebra containing $M.$}
Certainly each $\bar{a}$ is affiliated with $\A_{\omega},$ and it is the smallest von Neumann algebra on $\H_{\omega}$ with this property.

It will occasionally be useful to describe $\A_{\omega}$ in terms of the ``commutator'' definition of affiliation from section \ref{app:unbounded-operators}, instead of the ``polar decomposition'' definition exploited in the preceding paragraph.
We said in section \ref{app:unbounded-operators} that a bounded operator $L'$ commutes strongly with an unbounded, closed operator $T$ if $L'$ maps $D_T$ to itself and satisfies $T L'|_{D_T} = L' T.$
We would like to define $\A_0'$ to be the set of all bounded operators acting on $\H_{\omega}$ that commute strongly with all closures of operators in $\A_0,$ and define $\A_{\omega}$ to be the commutant $(\A_0')'.$
However this does not exactly work as stated, because $\A_0'$ is not necessarily closed under adjoints, which means that $(\A_0')'$ is not necessarily a von Neumann algebra.
There are two ways we could fix this: before taking the commutant, we could replace $\A_0'$ with the \textit{biggest} adjoint-closed set contained in $\A_0',$ or the \textit{smallest} adjoint-closed set containing $\A_0'.$
What we really want is for $\A_{\omega}$ to be the smallest von Neumann algebra for which $\A_0$ is affiliated, which means we want $\A_{\omega}'$ to be the biggest von Neumann algebra contained in $\A_0'.$
So we define
\begin{equation} \label{eq:a-omega-commutant-def}
	\A_{\omega} = (\A_0' \cap (\A_0')^{\dagger})'.
\end{equation}
The utility of this definition as compared to the preceding one is that it allows us to apply the usual logic of von Neumann's double commutant theorem --- to check that something is contained in $\A_{\omega},$ we need only check that it commutes with everything in $\A_0' \cap (\A_0')^{\dagger}.$

To close this subsection, we will show that $\A_0' \cap (\A_0')^{\dagger}$ is itself a von Neumann algebra, and that the two characterizations given above of $\A_{\omega}$ are actually the same.
First note that for a closed operator $T$ that commutes strongly with the bounded operator $L'$, one has for $|\chi\rangle \in D_{T^{\dagger}}$ and $|\psi\rangle \in D_{T}$ the formula
\begin{equation}
	\langle (L')^{\dagger} \chi | T \psi\rangle
		= \langle \chi | L' T \psi\rangle 
		= \langle \chi | T L' \psi \rangle
		= \langle (L')^{\dagger} T^{\dagger} \chi | \psi\rangle.
\end{equation}
This is a bounded function of $|\psi\rangle,$ so by section \ref{app:unbounded-operators} we can conclude that $(L')^{\dagger} |\chi\rangle$ is in $D_{T^{\dagger}}$ and $(L')^{\dagger}$ commutes with $T^{\dagger}$ on that domain.
Flipping the argument around, we see that $L'$ commutes strongly with $T$ if and only if $(L')^{\dagger}$ commutes strongly with $T^{\dagger}.$

Now we are ready to show that $\A_0' \cap (\A_0')^{\dagger}$ is itself a von Neumann algebra, and is equal to $\A_{\omega}'.$
It is clearly closed under adjoints, sums, and finite products, so to establish that it is a von Neumann algebra we need only show that it is closed under strong limits.
This is not so hard; if $L_n'$ is a sequence of operators in $\A_0'$ that converges strongly to $L'$, then for any $a \in \A_0$ and any $|\psi\rangle$ in the domain of $a$ one has 
\begin{equation}
	a L_n' |\psi\rangle = L_n' a |\psi\rangle \to L' a |\psi\rangle,
\end{equation}
and since $a$ is a closed operator this gives $L' |\psi\rangle$ in the domain of $a$ with $[a, L'] |\psi\rangle = 0$; so $L'$ commutes strongly with everything in $\A_0,$ and we conclude that $\A_0'$ is closed under strong limits.
The characterization of $(\A_0')^{\dagger}$ as the set of operators that commute strongly with all adjoints of operators in $\A_0$ establishes, by the same argument we have just given, that $(\A_0')^{\dagger}$ is closed under strong limits, and from this we conclude
\begin{equation}
	\A_{\omega}' = \A_0' \cap (\A_0')^{\dagger}.
\end{equation}

Finally, we show that the characterization of $\A_{\omega}$ from equation \eqref{eq:a-omega-commutant-def} is equivalent to the one written in terms of polar decompositions.
If $L'$ is in $\A_{\omega}',$ then by the above discussion it commutes strongly with both $T$ and $T^{\dagger},$ hence with $T^{\dagger} T.$
The details of the spectral theorem (see e.g. \cite[chapter 13]{Rudin:functional-book}) then imply that they commute with all bounded functions of $|T|.$
Given the polar decomposition
\begin{equation}
	T = V_T |T|,
\end{equation}
one then sees that $L'$ and $(L')^{\dagger}$ must commute with $V_T,$ for we have
\begin{equation}
	V_T L'|T|
		= V_T |T| L'|_{D_{|T|}}
		= T L'|_{D_T}
		= L' T
		= L' V_T |T|.
\end{equation}
So $L'$ and $V_T$ commute on the closure of the image of $|T|$.
As for the orthocomplement of the image of $|T|$, if $|\chi\rangle$ is a vector in this orthocomplement, then we have for each $|\psi\rangle \in D_{|T|}$ the formula
\begin{equation}
	\langle L' \chi | |T| \psi\rangle 
		= \langle \chi | (L')^{\dagger} |T| \psi\rangle
		= \langle \chi | |T| (L')^{\dagger} \psi\rangle
		= 0.
\end{equation}
So $L'$ preserves the orthocomplement of the image of $|T|$, but $V_T$ vanishes on this set, so on this set we have $L' V_T = V_T L' = 0.$
We therefore have $[L', V_T] = 0$, and similar considerations work for the adjoint $[L', V_T^{\dagger}]=0.$
This tells us that $L'$ commutes with $|T|, V_T,$ and $V_T^{\dagger}$, which completes the proof that the ``commutator definition'' of $\A_{\omega}$ is no weaker than the ``polar decomposition'' definition: everything in $\A_0' \cap (A_0')^{\dagger}$ commutes with everything made up of polar decompositions coming from $\A_0,$ so we have 
\begin{equation}
	(\A_0' \cap (\A_0')^{\dagger})'
		\supseteq \text{(minimum vN algebra generated by polar decompositions)}.
\end{equation}

For the converse inclusion, one needs to show that if an operator commutes with the minimum von Neumann algebra generated by polar decompositions, then it is in $\A_0' \cap (\A_0')^{\dagger}.$
This direction is not particularly hard.
Given a closed operator $T$ with polar decomposition
\begin{equation}
	T = V_T |T|,
\end{equation}
one can project onto bounded subsets of the spectrum of $|T|$ to write $T$ as a limit
\begin{equation}
	T = \lim_n V_T |T|_n,
\end{equation}
with each $|T|_n$ bounded, and where this limit converges whenever applied to any vector in the domain of $T$.
Anything that commutes with $V_T$ and $|T|_n$ is then easily shown to commute strongly with $T$.
This tells us that anything commuting with the minimal von Neumann algebra constructed from $V_T$ and bounded functions of $|T|$, necessarily commutes strongly with $T$.
But this minimal von Neumann algebra also necessarily contains $V_T^{\dagger}$ (since von Neumann algebras are closed under adjoints), and this gives a similar proof of strong commutation with $T^{\dagger} = |T| V_T^{\dagger}.$

\subsection{Cyclicity}
\label{app:cyclicity}

By construction, the set
\begin{equation}
	\{a | \omega\rangle \,|\, a \in \A_0\}
\end{equation}
is dense in $\H_{\omega}.$
We will show that the same is true for the set
\begin{equation}
	\{L | \omega\rangle \,|\, L \in \A_\omega\}.
\end{equation}
To see this, it suffices to show that every vector in the first set can be approximated by vectors in the second set.
But using the polar decomposition (and abusing notation by writing $a$ for the closure $\bar{a}$), we can write
\begin{equation}
	a |\omega\rangle = V_{a} |a| |\omega\rangle.
\end{equation}
The operator $|a|$ can be approximated on its domain as a sequence of bounded operators $|a|_n$ obtained by projecting $|a|$ onto the spectral subset $[0, n].$
So we have
\begin{equation}
	a |\omega\rangle = \lim_{n} V_{a} |a|_n |\omega\rangle,
\end{equation}
and the operator $V_{a} |a|_n$ is in $\A_{\omega}.$

\subsection{Domains and the commutant in the cyclic-separating case}
\label{app:commutant-domain}

Let $a$ and $b$ be elements of the $*$-algebra $\A_0,$ treated as closed, unbounded operators acting on the Hilbert space $\H_{\omega}.$
In the cyclic-separating case, $J_{\omega}$ maps $\A_{\omega}$ to $\A_{\omega}'.$
We would like to use this to argue that in the cyclic-separating setting, $J_{\omega} b|\omega\rangle$ is in the domain of $a.$

Since $a$ is a closed operator, it suffices to find a sequence of vectors $|\psi_n\rangle$ approaching $J_{\omega} b|\omega \rangle$ such that each $|\psi_n\rangle$ is in the domain of $a$, and such that the sequence $a |\psi_n\rangle$ also converges.
We will write the polar decomposition of $b$ as $V_b |b|,$ and the projection of $|b|$ onto the $[0,n]$ spectral subspace as $|b|_n.$
If we choose
\begin{equation}
	|\psi_n\rangle = J_{\omega} V_{b} |b|_n |\omega\rangle,
\end{equation}
then we clearly have $|\psi_n\rangle \to J b |\omega\rangle.$
But each $V_b |b|_n$ is in $\A_{\omega},$ so $J_{\omega} V_{b} |b|_n J$ is in $\A_{\omega}'.$
Because $a$ is affiliated to $\A_{\omega},$ it follows that $|\psi_n\rangle$ is in the domain of $a,$ and we have
\begin{equation}
	a |\psi_n\rangle
		= (J_{\omega} V_b |b|_n J_{\omega}) a |\omega\rangle.
\end{equation}
We need to show that this is a Cauchy sequence.
For this, we write
\begin{equation}
	\left\lVert a |\psi_n\rangle - a |\psi_m\rangle \right\rVert^2
		= \Big\langle (J_{\omega} V_b (|b|_n - |b|_m) J_{\omega}) a \omega \Big| (J_{\omega} V_b (|b|_n - |b|_m) J_{\omega}) a \omega \Big\rangle.
\end{equation}
We can move the bounded operators from the bra to the ket, commute it past $a$, then move $a$ from the ket to the bra.
This yields
\begin{equation}
	\left\lVert a |\psi_n\rangle - a|\psi_m\rangle \right\rVert^2
	= \Big\langle a^* a \omega \Big| J_{\omega} (|b|_n - |b|_m)^2  \omega \Big\rangle.
\end{equation}
Without loss of generality, we take $n < m,$ so that we have $|b|_n |b|_m = |b|_n^2.$
This gives
\begin{equation}
	\left\lVert a |\psi_n\rangle - a|\psi_m\rangle \right\rVert^2
	= \Big\langle a^* a \omega \Big| J_{\omega} (|b|_m^2 - |b|_n^2)  \omega \Big\rangle.
\end{equation}
The vector $|\omega\rangle$ is in the domain of $b^{\dagger} b = |b|^2,$ so the spectral theorem tells us that $|b|^2_n|\omega\rangle$ is a Cauchy sequence, and therefore $\lVert a |\psi_n\rangle$ is a Cauchy sequence.
This completes the proof that $J_{\omega} b |\omega\rangle$ is in the domain of $a.$
The calculation also tells us that $J_{\omega} a |\omega\rangle$ is in the domain of $b,$
and gives us the formula
\begin{equation} \label{eq:unbounded-commuting}
	a J_{\omega} b |\omega\rangle
		= J_{\omega} b J_{\omega} a |\omega\rangle.
\end{equation}

\subsection{A purity theorem}
\label{app:purity-theorem}

An algebraic state on a $*$-algebra $\A_0$ is said to be \textit{pure} if it cannot be written nontrivially as a convex sum of distinct states,
\begin{equation} \label{eq:app-mixed-formula}
	\omega = p \omega_1 + (1-p) \omega_2, \qquad 0 < p < 1, \quad \omega_j \neq \omega.
\end{equation}
For the purposes of this appendix, we will say that a state with this property is pure in sense (i).
This is a good mathematical definition, and it accounts for the notion of a state being mixed if it can be represented as a probabilistic mixture, but it is missing something of the information-theoretic understanding of purity: we would really like to say that a state is pure if it cannot be nontrivially entangled with a reference system.
For this, we introduce a second notion of purity: we say that $\omega$ is pure in sense (ii) if and only if, in the GNS representation, the von Neumann algebra $\A_{\omega}$ is equal to the full algebra of bounded operators on $\H_{\omega}$.
When $\omega$ is pure in sense (ii), the GNS construction introduces no new degrees of freedom to purify $\omega.$

Type (i) purity and type (ii) purity are equivalent for GNS representations of C$^*$-algebras; this was proved by Segal in \cite{Segal:irreducible}.
However, it turns out that they are not equivalent for general GNS representations of $*$-algebras.
Here we show --- as in e.g. \cite[theorem 6.3]{Powers:self-adjoint} --- that a state is pure in sense (i) if and only if its \textit{weak commutant} is trivial.
The weak commutant is generally larger than the set we called $\A_0'$ in section \ref{app:vN}, which is often called the strong commutant \cite{Powers:self-adjoint, Borchers:weak-commutant}.\footnote{See also \cite[p.101]{Streater:book} for a discussion of commutativity and reducibility in the context of the Wightman axioms.}
So if a state is pure in sense (i) then it is pure in sense (ii), but the converse is not necessarily true.

We said that for a $*$-algebra $\A_0$ in a GNS representation $\H_{\omega}$, the bounded operator $L'$ is in the strong commutant if we have
\begin{equation}
	L' a \subseteq a L' \qquad \forall\, a \in \A_0.
\end{equation}
We say it is in the weak commutant if instead we only have
\begin{equation}
	L' a \subseteq (a^*)^{\dagger} L' \qquad \forall\, a \in \A_0.
\end{equation}
Because $(a^*)^{\dagger}$ may be a nontrivial extension of $a,$ these two definitions are not equivalent.
One may also say that $L'$ is in the weak commutant if we have
\begin{equation}
	\langle b \omega | L' | a c \omega\rangle
		= \langle a^* b \omega | L' | c \omega \rangle \qquad \forall\, a,b,c\in\A_0.
\end{equation}

Now, suppose that $\omega$ is mixed in sense (i), and has a decomposition of the form \eqref{eq:app-mixed-formula}.
The map $\psi \equiv p \omega_1$ is a positive linear functional on $\A_0$ that is bounded above by $\omega.$
From this (and the Cauchy-Schwarz inequality for $\psi$) one can see that with respect to the GNS representation, we have
\begin{equation}
	|\psi(a^* b)|
		\leq \sqrt{\psi(a^* a)} \sqrt{\psi(b^* b)}
		\leq \sqrt{\omega(a^* a)} \sqrt{\omega(b^* b)}
		= \lVert a |\omega\rangle \rVert \lVert b |\omega \rangle \rVert.
\end{equation}
So the map
\begin{equation}
	\left(a |\omega\rangle, b |\omega\rangle \right) \mapsto \psi(a^* b)
\end{equation}
is linear in $b$ and antilinear in $a,$ and is bounded in terms of the $|\omega\rangle$ inner product.
By the Riesz lemma, there exists a bounded operator $\hat{\psi}$ on $\H_{\omega}$ with matrix elements
\begin{equation}
	\langle a \omega | \hat{\psi} | b \omega\rangle = \psi(a^* b).
\end{equation}
This is clearly a positive operator, and is bounded above by one.
Because we have $p \omega_1 \neq \omega,$ we also have $\hat{\psi} \neq 1.$
But we have
\begin{equation}
	\langle b \omega | \hat{\psi} | a c \omega \rangle
		= \psi(b^* a c) = \psi((a^* b)^* c) = \langle a^* b \omega | \hat{\psi} | c \omega \rangle,
\end{equation}
so $\hat{\psi}$ is a nontrivial operator in the weak commutant of $\A_0.$

For the converse, suppose that the weak commutant of $\A_0$ is nontrivial.
Let $L'$ be a nontrivial, bounded operator that commutes weakly with $\A_0.$
It is not hard to show that $(L')^{\dagger} L'$ is a positive operator that commutes weakly with $\A_0,$ and by rescaling it we can find an operator in the weak commutant that is bounded between $0$ and $1.$
Call this operator $\hat{\psi},$ and define
\begin{equation}
	\psi(a)
		= \langle \omega | \hat{\psi} | a \omega\rangle.
\end{equation}
It is easy to show that this is a positive, linear functional on $\A_0$ that is bounded above by $\omega,$ so the decomposition
\begin{equation}
	\omega = \psi + (1-\psi)
\end{equation}
implies that $\omega$ is mixed in sense (i).

\section{Some lesser-known aspects of modular theory}

In the main text, we need to employ a few aspects of modular theory that are well known to experts but not as widely diffused as the basic theory.
These are (i) the construction of modular operators for non-separating states, (ii) the centralizer theorem, and (iii) purification via the natural cone.
This appendix reviews these notions in a manner that is less complete than in textbook accounts, but which does include some details of the relevant proofs.

For other references on point (i) see \cite[appendix C]{Araki:Lp} or \cite[appendix A]{Ceyhan:QNEC}.
For other references on point (ii) see \cite[chapter VIII.2]{Takeaski:II} or the interesting recent application \cite{deBoer:berry}.
For other references on point (iii) see the original papers \cite{Araki:natural-1, Araki:natural-2} or the textbook account \cite[sections 10.23-10.25]{Stratila:book}.
Background on cyclic-separating modular theory can be found in \cite{Witten:notes, Sorce:modular}.

As for the previous appendix, because the material is technical, we begin with a bulleted summary of key points.

\subsection{Key points}

\begin{itemize}
	\item If a state $|\omega\rangle$ and von Neumann algebra $\A$ are obtained from a GNS representation, then $|\omega\rangle$ is always cyclic, but it may not be separating.
	\item For a non-separating state $|\omega\rangle,$ the projection $P_{\omega}$ that maps onto the closure of $\A' |\omega\rangle$ is a nontrivial operator in $\A$.
	One defines the Tomita operator $S_{\omega}$ as the closure of
	\begin{equation}
		S_{\omega} a |\omega\rangle = P_{\omega} a^{\dagger} |\omega\rangle.
	\end{equation}
	Its polar decomposition is $J_{\omega} \Delta_{\omega}^{1/2}.$
	The usual KMS relation generalizes to
	\begin{equation}
			\langle \Delta_{\omega}^{1/2} a \omega | \Delta_{\omega}^{1/2} b \omega \rangle
			= \langle b^{\dagger} \omega | P_{\omega} | a^{\dagger} \omega\rangle.
		\end{equation}
	The modular conjugation $J_{\omega}$ is an antilinear partial isometry whose support and image are both equal to the support of $P_{\omega}.$
	It satisfies
	\begin{equation}
		J_{\omega}^2 = P_{\omega}.
	\end{equation}
	The kernel of $\Delta_{\omega}^{1/2}$ is the orthocomplement of the support of $P_{\omega}.$
	One still has the usual formula
	\begin{equation}
		J_{\omega} \Delta_{\omega} J_{\omega} = \Delta_{\omega}^{-1}
	\end{equation}
	provided that one defines the inverse to be zero on the kernel of $\Delta_{\omega}.$
	\item In the cyclic-but-not-separating case, we do not have the usual relation $J_{\omega} \A J_{\omega} = \A'.$
	Instead we have
	\begin{equation}
		J_{\omega} \A J_{\omega} = P_{\omega} \A' P_{\omega}.
	\end{equation}
	\item Given a cyclic-separating state $|\omega\rangle$ for a von Neumann algebra $\A,$ an operator $a \in \A$ is said to be in the \textit{centralizer} if it is left fixed by the modular flow.
	This turns out to be the same as the identity
	\begin{equation}
		\langle \omega | [a, b] |\omega\rangle = 0 \qquad \forall\, b \in \A.
	\end{equation}
	\item
	For any $a$ in the centralizer, one has
	\begin{equation}
		\Delta_{\omega}^{z} a |\omega\rangle
			= a |\omega\rangle
	\end{equation}
	for any complex power $z$ of the modular operator.
	\item Given a cyclic-separating state $|\omega\rangle,$ there is a special set of vectors in $\H$ called the \textit{natural cone},
	\begin{equation}
		\mathfrak{P}_{1/4}
			= \bar{\{ \Delta_{\omega}^{1/4} p |\omega\rangle \,|\, p \in \A, p \text{ positive}\}}.
	\end{equation}
	Every ultraweakly continuous functional of $\A$ has a unique vector representative in $\mathfrak{P}_{1/4}.$
	\item To check that a cyclic-separating vector $|\psi\rangle$ is in $\mathfrak{P}_{1/4},$ it suffices to check that its modular conjugation agrees with the modular conjugation of $|\omega\rangle,$ and that one has
	\begin{equation}
		\langle \psi | Z | \omega \rangle \geq 0
	\end{equation}
	for any positive $Z$ in the center of $\A$.
\end{itemize}

\subsection{Non-separating modular theory}
\label{app:modular}

Suppose we have a Hilbert space $\H$ with a von Neumann algebra $\A$, and a vector $|\omega\rangle$ that is cyclic for $\A$.
This means that the vectors
\begin{equation}
	\{a |\omega\rangle\,|\, a \in \A\}
\end{equation}
form a dense subspace of $\H$.
As discussed in appendix \ref{app:cyclicity}, this is always true for the GNS vector and the von Neumann algebra induced by a state on a general $*$-algebra.
One says that $|\omega\rangle$ is separating for $\A$ if it is cyclic for the commutant algebra $\A'.$
However, this is not generically true in the GNS setting --- it is true only if the algebraic state $\omega$ is faithful in the sense that
\begin{equation}
	(\omega(a^* a) = 0) \quad \Rightarrow \quad (a=0).
\end{equation}
In the language of finite-dimensional quantum mechanics, a state fails to be separating if its density matrix is not full rank.
We do not want to exclude ourselves from constructing canonical purifications for such states, so we must use modular theory in the generic, separating case.\footnote{Modular theory can also be generalized to the non-cyclic case; for a nice account see \cite[appendix A]{Ceyhan:QNEC}.}
We will freely make use of notions involving unbounded operators, so the reader may wish to consult appendix \ref{app:unbounded-operators} before proceeding.

When a state $|\omega\rangle$ is cyclic but not separating for the von Neumann algebra $\A$, the subspace $\A'|\omega\rangle$ is not dense in Hilbert space.
We define
\begin{equation}
	P_{\omega}
		\equiv P_{\bar{\A'|\omega\rangle}}
\end{equation}
to be the orthogonal projection onto the closure.
This operator $P_{\omega}$ can be shown to live in $\A.$\footnote{By definition of $P_{\omega}$, for $a' \in \A'$ one has $a' P_{\omega} = P_{\omega} a' P_{\omega}.$ This gives
$$
	a' P_{\omega} = P_{\omega} a' P_{\omega} = (P_{\omega} (a')^{\dagger} P_{\omega})^{\dagger} = ((a')^{\dagger} P_{\omega})^{\dagger} = P_{\omega} a',
$$
and $P_{\omega}$ is in $\A$ by the double commutant theorem.}

One defines the antilinear ``pre-Tomita operator'' $S_{\omega, 0}$ on the domain $\A |\omega\rangle$ via the formula
\begin{equation}
	S_{\omega, 0} a |\omega\rangle
		= P_{\omega} a^{\dagger} |\omega\rangle.
\end{equation}
One also defines the operator $F_{\omega, 0}$ to be zero on the complement of $P_{\omega} \H,$ and to act on $\A' |\omega\rangle$ as
\begin{equation}
	F_{\omega, 0} a' |\omega\rangle
		= (a')^{\dagger} |\omega\rangle.
\end{equation}
One easily checks that for $|\psi\rangle$ in the domain of $F_{\omega, 0},$ we have
\begin{equation}
	\langle a \omega | F_{\omega, 0} \psi \rangle = \langle \psi | S_{\omega, 0} a \omega \rangle.
\end{equation}
For antilinear operators, this is the appropriate equation to show that $F_{\omega, 0}^{\dagger}$ is an extension of $S_{\omega, 0},$ so $S_{\omega, 0}$ is closed.
Its closure is called $S_{\omega}$, and its polar decomposition is written
\begin{equation}
	S_{\omega}
		= J_{\omega} \Delta_{\omega}^{1/2}.
\end{equation}
The antilinear partial isometry $J_{\omega}$ is called the modular conjugation, and the positive operator $\Delta_{\omega}$ is called the modular operator.
From this decomposition one can compute
\begin{equation}
	\langle \Delta_{\omega}^{1/2} a \omega | \Delta_{\omega}^{1/2} b \omega \rangle
		= \langle b^{\dagger} \omega | P_{\omega} | a^{\dagger} \omega\rangle,
\end{equation}
which generalizes the usual KMS relation to the non-separating setting.

Both $J_{\omega}$ and $\Delta_{\omega}^{1/2}$ will generally have a kernel.
The kernel of $\Delta_{\omega}^{1/2}$ is the same as the kernel of $S_{\omega},$ and this is easily seen to be the orthocomplement $(P_{\omega} \H)^{\perp}$ via the manipulations
\begin{align}
	\begin{split}
		0 = S_{\omega} a |\omega\rangle \quad
		& \Leftrightarrow \quad 0 = P_{\omega} a^{\dagger} |\omega\rangle \\
		& \Leftrightarrow \quad \langle a' \omega | a^{\dagger} \omega \rangle = 0 \quad \forall\, a' \in \A' \\
		& \Leftrightarrow \quad \langle a \omega | (a')^{\dagger} \omega \rangle = 0 \quad \forall\, a' \in \A' \\
		& \Leftrightarrow \quad \langle a \omega | a' \omega \rangle = 0 \quad \forall\, a' \in \A'.
	\end{split}
\end{align}
To find the kernel of $J_{\omega},$ we note that the partial isometry is defined so that its image is the closure of the image of $S_{\omega},$ and its support is the closure of the image of $\Delta_{\omega}^{1/2}.$
Since $\Delta_{\omega}^{1/2}$ is self-adjoint, the closure of its image is the orthocomplement of its kernel.
So we have
\begin{equation}
	\im{J_{\omega}} = \bar{\A'|\omega\rangle} = P_{\omega} \H
\end{equation}
and
\begin{equation}
	\ker{J_{\omega}} = \ker{\Delta_{\omega}^{1/2}} = (P_{\omega} \H)^{\perp}
\end{equation}
So $J_{\omega}$ is a partial isometry from $P_{\omega} \H$ onto itself.
This gives $J_{\omega} J_{\omega}^{\dagger} = J_{\omega}^{\dagger} J_{\omega} = P_{\omega}$

To recover the usual identity $J_{\omega}^{\dagger} = J_{\omega},$ we note that because $P_{\omega}$ is in $\A$, we have
\begin{equation}
	S_{\omega} S_{\omega} a |\omega\rangle 
		= S_{\omega} P_{\omega} a^{\dagger} |\omega\rangle
		= P_{\omega} a P_{\omega} |\omega\rangle 
		= P_{\omega} a |\omega\rangle.
\end{equation}
So on its domain, we have $S_{\omega}^2 = P_{\omega}.$
This gives the formula
\begin{equation}
	P_{\omega}|_{D_{S_{\omega}}}
		= J_{\omega} \Delta_{\omega}^{1/2} J_{\omega} \Delta_{\omega}^{1/2}.
\end{equation}
So the right-hand side equals zero on the part of the domain of $S_{\omega}$ that is in the kernel, and equals the identity on the complement.
If we choose to define $\Delta_{\omega}^{-1/2}$ by setting it equal to zero on the kernel of $\Delta_{\omega}^{1/2}$ and inverting it normally elsewhere, then the above formula gives
\begin{equation}
	J_{\omega} \Delta_{\omega}^{1/2} J_{\omega} \supseteq \Delta_{\omega}^{-1/2}.
\end{equation}
To show equality, one must show that anything in the domain of the left is in the domain of the right.
In other words, anything in the domain of $J_{\omega} \Delta_{\omega}^{1/2} J_{\omega}$ must be in the image of $\Delta_{\omega}^{1/2}.$
For any such vector $|\psi\rangle,$ the vector $J_{\omega} |\psi\rangle$ is in the domain of $\Delta_{\omega}^{1/2},$ i.e., it is in the domain of $S_{\omega}.$
However, $J_{\omega} |\psi\rangle$ is \textit{also} in the image of $P_{\omega}.$
Since it is in the domain of $P_{\omega}$ and in the image of $S_{\omega},$ the identity $S_{\omega}^2 \subseteq P_{\omega}$ implies that $J_{\omega}|\psi\rangle$ is in the image of $S_{\omega}.$ So we can write
\begin{equation}
	J_{\omega} |\psi\rangle = S_{\omega} |\chi\rangle = J_{\omega} \Delta_{\omega}^{1/2} |\chi\rangle,
\end{equation}
and left-multiplying by $J_{\omega}^{\dagger}$ gives
\begin{equation}
	|\psi\rangle = P_{\omega} \Delta_{\omega}^{1/2} |\chi\rangle.
\end{equation}
But we already showed that the image of $\Delta_{\omega}^{1/2}$ is left fixed by $P_{\omega},$ so we have completed the proof of equality in
\begin{equation}
	J_{\omega} \Delta_{\omega}^{1/2} J_{\omega} = \Delta_{\omega}^{-1/2}.
\end{equation}
One may then write
\begin{equation}
	\Delta_{\omega}^{-1/2}
		= J_{\omega} \Delta_{\omega}^{1/2} J_{\omega} = J_{\omega} J_{\omega} J_{\omega}^{\dagger} \Delta_{\omega}^{1/2} J_{\omega},
\end{equation}
and uniqueness of the polar decomposition implies that $J_{\omega}^2$ must be equal to the projector onto the closure of the image of $\Delta_{\omega}^{-1/2},$ which gives
\begin{equation}
	J_{\omega}^2 = P_{\omega},
\end{equation}
hence $J_{\omega} = J_{\omega}^{\dagger}.$

Readers may be familiar with the fact that in the cyclic-separating case, conjugating by $J_{\omega}$ maps $\A$ to $\A'.$
This is not true in the non-separating case.
The culprit is the incomplete support of $J_{\omega}.$
The general formula is
\begin{equation} \label{eq:general-conjugation-algebra}
	J_{\omega} \A J_{\omega} = P_{\omega} \A' P_{\omega},
\end{equation}
To prove this, one can use the ``separating trick'' from section \ref{sec:separating-extension} to construct an auxiliary space $\H_{\phi}$ carrying a von Neumann algebra $\A_{\phi}$ isomorphic to $\A,$ with a state $|\phi\rangle$ for which $\A_{\phi}$ is cyclic and separating, and that has a subspace isometric to $\H$ via the map
\begin{equation}
	V : a |\omega\rangle \mapsto \sqrt{2} a P_{\omega} |\phi\rangle.
\end{equation}
From equation \eqref{eq:separating-mod-op}, one has
\begin{equation}
	J_{\omega} = V^{\dagger} P_{\omega} J_{\phi} P_{\omega} V,
\end{equation}
hence
\begin{equation} \label{eq:long-modular-conj}
	J_{\omega} a J_{\omega}
		= (V^{\dagger} P_{\omega} J_{\phi} P_{\omega} V) a (V^{\dagger} P_{\omega} J_{\phi} P_{\omega} V)
\end{equation}
To simplify this expression, one uses that the projector $P_{\phi}',$ which projects onto the image of $V$, is in $\A_{\phi}'.$
From this one can deduce $V a V^{\dagger} = P_{\phi}' a P_{\phi}'$, since one only needs to check matrix elements in states in the image of $V$, and one finds
\begin{equation}
	\langle V b \omega | P_{\phi}' a P_{\phi}' |V c \omega \rangle
		= 2 \langle b P_{\omega} \phi | P_{\phi}' a P_{\phi}' | c P_{\omega} \phi\rangle
		= 2 \langle b P_{\omega} P_{\phi}' \phi | a | c P_{\omega} \phi\rangle
\end{equation}
and then one can use (from section \ref{sec:separating-extension}) the identity $P_{\omega} P_{\phi}' |\phi\rangle = P_{\omega} |\phi\rangle$ to deduce
\begin{equation}
	\langle V b \omega | P_{\phi}' a P_{\phi}' |V c \omega \rangle
	= 2 \langle b P_{\omega} \phi | a | c P_{\omega} \phi\rangle
	= \langle V b \omega | V a c \omega \rangle
	= \langle V b \omega | V a V^{\dagger} | V c \omega\rangle. 
\end{equation}
Plugging $V a V^{\dagger} = P_{\phi}' a P_{\phi}'$ into equation \eqref{eq:long-modular-conj}, one finds
\begin{equation}
	J_{\omega} a J_{\omega}
	= V^{\dagger} P_{\omega} (J_{\phi} a J_{\phi}) P_{\omega} V,
\end{equation}
where we have simplified using $P_{\phi}' J_{\phi} = J_{\phi} P_{\omega}$ and $P_{\phi}' V = V.$
Finally, one can use the definition of $V$ to deduce $P_{\omega} V = V P_{\omega}$, which gives
\begin{equation}
	J_{\omega} a J_{\omega}
	\in P_{\omega} V^{\dagger} \A_{\phi}' V P_{\omega}.
\end{equation}
For any $a' \in \A_{\phi}'$ it is straightforward to check $V^{\dagger} a' V \in \A'$:
\begin{align}
	\begin{split}
	\langle c_1 \omega| (V^{\dagger} a' V) b | c_2 \omega\rangle
		& = 2 \langle c_1 P_{\omega} \phi | a' | b c_2 P_{\omega} \phi \rangle \\
		& = 2 \langle b^{\dagger} c_1 P_{\omega} \phi | a' | c_2 P_{\omega} \phi\rangle \\
		& = \langle V b^{\dagger} c_1 \omega | a' V c_2 \omega \rangle \\
		& = \langle c_1 \omega | b (V^{\dagger} a' V) |c_2 \omega\rangle.
	\end{split} 
\end{align}
This completes the proof of equation \eqref{eq:general-conjugation-algebra}.

\subsection{The centralizer theorem}
\label{app:centralizer}

In many (but not all) physical contexts, one consider von Neumann algebras that are factors, $\A \cap \A' = \comps1.$
When this equation is satisfied, one says that the center is trivial.
Given a state $|\omega\rangle,$ one can consider the set of operators that are ``effectively in the center as far as $|\omega\rangle$ is concerned.''
This is the set
\begin{equation}
	\Z_{\omega}
		= \{a \in \A \,|\, \langle \omega | [a, b] | \omega\rangle = 0 \text{ for all } b \in \A\}.
\end{equation}
One can also consider the modular flow $\Delta_{\omega}^{it}$ and study the set of operators that are fixed under this flow:
\begin{equation}
	\mathcal{C}_{\omega}
		= \{ a \in \A \,|\, \Delta_{\omega}^{-it} a \Delta_{\omega}^{it} = a \text{ for all } t \in \mathbb{R}\}.
\end{equation}
This set is called the centralizer of $|\omega\rangle$ in $\A.$
A very useful theorem says that in the case of a cyclic and separating state, these two notions are the same: $\Z_{\omega} = \mathcal{C}_{\omega}.$\footnote{We do not know whether the identity $\Z_{\omega} = \mathcal{C}_{\omega}$ holds outside of the cyclic-separating case. We expect it does, but we have been unable to find any discussion of this in the literature. In any case, in the main text we only apply this theorem to the cyclic-separating case.}
We will sketch the proof below.
A corollary we will use in the main text is that whenever we have $a \in \Z_{\omega},$ the vector $a |\omega\rangle$ is in the domain of every power of $\Delta_{\omega},$ and satisfies $\Delta_{\omega}^{z} a |\omega\rangle = a |\omega\rangle.$
This follows from the fact that when $a$ is in $\mathcal{C}_\omega,$ the function
\begin{equation}
	t \mapsto \Delta_{\omega}^{it} a |\omega\rangle
\end{equation}
is constant for real $t,$ so it admits a constant analytic continuation to arbitrary values of $t.$
For the connection between analytic continuations and the domains of positive operators, see \cite[section 2.4]{Sorce:modular}.

The proof of the relationship between $\Z_{\omega}$ and $\mathcal{C}_{\omega}$ works by studying analytic continuations of the modular flow.
As discussed in the preceding paragraph, when $a$ is in $\mathcal{C}_\omega,$ we have $\Delta_{\omega}^{z} a |\omega\rangle = a |\omega\rangle$ for any $z,$ which gives in particular
\begin{equation}
	\langle \omega | b a |\omega\rangle
		= \langle \omega | b \Delta_{\omega} a |\omega\rangle
		= \langle \omega | a b |\omega\rangle,
\end{equation}
where in the last equality we have used the KMS relation --- see for example \cite{Sorce:modular}.
This implies $\mathcal{C}_\omega \subseteq \Z_{\omega}.$

For the converse inclusion, one uses the fact that there is a dense set of states $b |\omega\rangle$ for which $\Delta_{\omega}^{z} b^{\dagger} |\omega\rangle$ is an analytic function --- see for example \cite[theorem 10.20]{Stratila:book} --- and writes
\begin{equation}
	\alpha(t)
		= \langle \omega | b \Delta_{\omega}^{-it} a |\omega\rangle
		= \langle \Delta_{\omega}^{\bar{z}} b^{\dagger} \omega | a \omega\rangle|_{z = -it}.
\end{equation}
Analytically continuing to $z=1-it$, one finds
\begin{equation}
	\alpha(t + i)
	= \langle \Delta_{\omega}^{1+it} b^{\dagger} \omega | a \omega\rangle
	= \langle \omega | a \Delta_{\omega}^{it} b | \omega \rangle,
\end{equation}
where again we have used the KMS relation.
Now using the assumption $a \in \Z_{\omega},$ we find
\begin{equation}
	\alpha(t + i)
	= \langle \omega | a \Delta_{\omega}^{it} b \Delta_{\omega}^{-it} | \omega \rangle
	= \langle \omega | \Delta_{\omega}^{it} b \Delta_{\omega}^{-it} a | \omega \rangle
	= \alpha(t).
\end{equation}
So the analytic function $\alpha(t)$ is periodic in imaginary time; it is not hard to show that it is bounded, so it follows that $\alpha(t)$ is constant.
This gives a dense set of states for which the overlap with $\Delta_{\omega}^{-it} a |\omega\rangle$ is the same as the overlap with $a |\omega\rangle$; hence we have
\begin{equation}
	\Delta_{\omega}^{-it} a |\omega\rangle
		= a |\omega\rangle.
\end{equation}
To get from this to the conclusion $a \in \mathcal{C}_{\omega}$, we use the fact that modular flow sends $\A$ to itself to write, for any $a' \in \A',$ the formula
\begin{equation}
	(\Delta_{\omega}^{-it} a \Delta_{\omega}^{it}) a' |\omega\rangle
			= a' (\Delta_{\omega}^{-it} a \Delta_{\omega}^{it}) |\omega\rangle
			= a' a |\omega\rangle = a a' |\omega\rangle.
\end{equation}
Since $\A' |\omega\rangle$ is dense, this proves $\Delta_{\omega}^{-it} a \Delta_{\omega}^{it} = a,$ as desired.

\subsection{Purification via the natural cone}
\label{app:natural-cone}

Suppose we have a von Neumann algebra $\A$ with a cyclic-separating state $|\omega\rangle.$
Suppose we have a density matrix $\rho$ on $\H$ that induces a set of expectation values on $\A$ via the formula\footnote{Equivalently, we may start with the linear functional and demand only that it is positive and ``ultraweakly continuous,'' since this is equivalent to the existence of a density matrix representative; see \cite[theorem 54.10]{conway2000course}.}
\begin{equation}
	\rho(a) = \tr(\rho a).
\end{equation}
Is there a vector state in $\H$ reproducing these expectation values?
I.e., is there a state $|\rho\rangle$ with
\begin{equation}
	\langle \rho | a | \rho\rangle = \rho(a) \qquad \forall\, a \in \A?
\end{equation}
The answer is always yes when $\A$ possesses a separating vector --- see for example \cite[corollary 5.24]{Stratila:book}.
But there is a huge ambiguity in how this state should be chosen, since one can act on $|\rho\rangle$ with any unitary $U'$ in $\A'$ and get a purification that is just as good.
In \cite{Araki:natural-1, Araki:natural-2}, Araki introduced a family of subsets of $\H$ such that in each of these subsets, every global density matrix has a single, unique representative.
These subsets are built off of the cyclic-separating state $|\omega\rangle$, so the resulting purifications inherit some of their properties from $|\omega\rangle.$

Concretely, given a von Neumann algebra $\A$, one writes $\A_+$ for the set of positive elements --- this is the set of operators $a$ with
\begin{equation}
	\langle \psi | a | \psi\rangle \geq 0 \qquad \forall\, |\psi\rangle \in \H.
\end{equation}
Equivalently, it is the set of elements of $\A$ that can be written as $a = b^{\dagger} b$ for some other $b \in \A.$
For any $\alpha \in [0, 1/2],$ Araki defines the set
\begin{equation}
	\mathfrak{P}_{\alpha}
		= \bar{\{\Delta_{\omega}^{\alpha} p |\omega\rangle \,|\, p \in \A_+\}}.
\end{equation}
This is a \textit{cone} --- it is preserved under addition and under multiplication by nonnegative numbers.

There are several interesting relations between these cones.
The first is that if one defines a set of cones with respect to the commutant algebra,
\begin{equation}
	\mathfrak{P}'_{\alpha}
	= \bar{\{(\Delta_{\omega}')^{\alpha} p' |\omega\rangle \,|\, p' \in \A'_+\}},
\end{equation}
then one has
\begin{equation}
	\mathfrak{P}_{\alpha} = \mathfrak{P}'_{1/2-\alpha}.
\end{equation}
To see this, one fixes a vector $|\psi\rangle$ in $\mathfrak{P}_{\alpha}$ and a set of vectors $\Delta_{\omega}^{\alpha} p_n |\omega\rangle$ converging to $|\psi\rangle.$
Using standard identities for the modular operator, one computes
\begin{align}
	\begin{split} 
	|\psi\rangle
		& \lim_n \Delta_{\omega}^{\alpha} p_n |\omega\rangle \\
		& = \lim_n \Delta_{\omega}^{\alpha - 1/2} \Delta_\omega^{1/2} p_n |\omega\rangle \\
		& = \lim_n \Delta_{\omega}^{\alpha - 1/2} J_{\omega} S_{\omega} p_n |\omega\rangle \\
		& = \lim_n \Delta_{\omega}^{\alpha - 1/2} J_{\omega} p_n |\omega\rangle \\
		& = \lim_n (\Delta'_{\omega})^{1/2-\alpha} (J_{\omega} p_n J_{\omega}) |\omega\rangle \\
	\end{split}
\end{align}
where in the last step we have used $\Delta_{\omega}' = \Delta_{\omega}^{-1}.$
This sequence of equations, together with the fact that $J_{\omega}$ conjugates $\A$ to $\A',$ gives $\mathfrak{P}_{\alpha} \subseteq \mathfrak{P}'_{1/2-\alpha}.$
The reverse inclusion follows by similar logic.

Another interesting relation between the cones is that $\mathfrak{P}_{\alpha}$ is the dual cone to $\mathfrak{P}_{1/2-\alpha}.$
The ``dual cone'' of a cone $\mathfrak{C}$ is the set $\mathfrak{C}^*$ of all vectors that have nonnegative overlap with vectors in $\mathfrak{C}.$
For $|\psi\rangle \in \mathfrak{P}_{\alpha}$ and $|\eta \rangle \in \mathfrak{P}_{1/2-\alpha},$ the results of the paragraph guarantee the existence of positive sequences $p_n \in \A_+$ and $p'_n \in \A'_+$ with
\begin{equation}
	|\psi\rangle = \lim_{n} \Delta_{\omega}^{\alpha} p_n |\omega\rangle
\end{equation}
and
\begin{equation}
	|\eta\rangle = \lim_{n} (\Delta_{\omega})^{\alpha} p'_n |\omega\rangle.
\end{equation}
Taking the overlap $\langle \psi| \eta\rangle$ and using $\Delta_{\omega}' = \Delta_{\omega}^{-1},$ one finds positivity via the relation
\begin{equation}
	\langle p_n \omega | p'_n \omega \rangle = \langle \sqrt{p_n} \omega | p'_n | \sqrt{p_n} \omega\rangle \geq 0.
\end{equation}
This calculation shows that $\mathfrak{P}_{1/2 - \alpha}$ is contained in the dual cone of $\mathfrak{P}_{\alpha},$ but it does not show equality of these sets.
The argument for equality is a little complicated, so we omit it.
See the textbook account \cite[chapter 10]{Stratila:book} for a readable account in the special cases $\alpha=0$ and $\alpha=1/4$; the proofs generalize in a straightforward way to all $\alpha \in [0, 1/2].$

The observation that $\mathfrak{P}_{\alpha}$ is dual to $\mathfrak{P}_{1/2-\alpha}$ immediately tell us that $\mathfrak{P}_{1/4}$ is special --- it is ``self-dual.''
Indeed, it is often called \textit{the} self-dual cone or the ``natural cone'' of $\A$ with respect to $|\omega\rangle.$
One can show that for any density matrix $\rho$ on $\H$ --- equivalently, for any ultraweakly continuous linear functional on $\A$ --- there exists a unique vector $|\rho_{1/4}\rangle \in \mathfrak{P}_{1/4}$ reproducing all the expectation values of $\rho$ on $\A$.
Moreover, there is a nice continuity relationship between $\rho$ and $|\rho_{1/4}\rangle$ in the form of the inequality
\begin{equation}
	\lVert |\rho_{1/4}\rangle - |\sigma_{1/4}\rangle \rVert^2
		\leq \lVert \rho - \sigma \rVert_1 \leq \lVert |\rho_{1/4}\rangle - |\sigma_{1/4}\rangle \rVert \times \lVert |\rho_{1/4}\rangle + |\sigma_{1/4}\rangle \rVert,
\end{equation}
with
\begin{equation}
	\lVert \rho - \sigma \rVert_1
		= \sup_{a \in \A, \lVert a \rVert_{\infty} \leq 1} |\rho(a) - \sigma(a)|.
\end{equation}
This makes it very convenient to study an algebraic state by passing to this representative in the natural cone.
For this reason, some sources have called $|\rho_{1/4}\rangle$ the ``canonical purification'' of $\rho$ with respect to $|\omega\rangle$.
For reasons explained in the main text, we think this is somewhat misleading terminology --- in particular, the canonical purification of a state should be defined without reference to another state.
However, in the main text we do show that there is a connection between the notions of canonical purification and natural purification in certain settings.
To accomplish this, we will use  one last fact: for every cyclic-separating state in the natural cone, the modular conjugation $J_{\rho_{1/4}}$ is the same as the modular conjugation $J_{\omega}.$
Moreover, it was proved by Araki (\cite[theorem 4-5]{Araki:natural-1}) that if $|\psi\rangle$ is a cyclic-separating state with $J_{\psi} = J_{\omega}$ and one has
\begin{equation}
	\langle \psi | Z | \omega \rangle \geq 0
\end{equation}
for every positive $Z$ in the center of $\A$, then $|\psi\rangle$ is in the natural cone of $|\omega\rangle.$

\section{The necessity of time-reversal}
\label{app:time-reversal}

In the main text, our canonical purification functional $\hat{\omega}$ was defined on the algebra $\A_0 \otimes \A_0^{\text{op}}.$
However, the formula
\begin{equation}
	\hat{\omega}(a \otimes b)
	= \langle a^{*} \omega | J_{\omega} |b^{*} \omega\rangle 
\end{equation}
could equally well be defined on $\A_0 \otimes \A_0$.
To emphasize the importance of time reversal in canonical purification, we present here a simple example where $\hat{\omega}$ is not positive on $\A_0 \otimes \A_0.$

Let $\H$ be a finite-dimensional Hilbert space, let $\A_0$ be the set of linear operators $\B(\H)$, and let $\omega$ be a state induced by a density matrix $\rho.$
In section \ref{sec:finite-dimension}, we computed the canonical purification formula as
\begin{equation}
	\hat{\omega}(a \otimes b)
	= \langle \sqrt{\rho} | (a \otimes b) |\sqrt{\rho}\rangle.
\end{equation}
If we define this as a functional on $\A_0 \otimes \A_0,$ then we have
\begin{equation} \label{eq:negative-sum}
	\sum_{j, k} \hat{\omega}((a_j^* \otimes b_j^*) (a_k \otimes b_k))
	= \sum_{j, k} \langle \sqrt{\rho} | (a_j^*  a_k \otimes b_j^* b_k) |\sqrt{\rho}\rangle.
\end{equation}
But the representation of $\B(\H) \otimes \B(\H)$ on $\H \otimes \H^*$ has a reversed multiplication rule in the second factor, so this can be rewritten as
\begin{equation}
	\sum_{j, k} \hat{\omega}((a_j^* \otimes b_j^*) (a_k \otimes b_k))
	= \sum_{j, k} \langle \sqrt{\rho} | (a_j^* \otimes b_k) (a_k \otimes b_j^*) |\sqrt{\rho}\rangle.
\end{equation}
We can see already why we should not expect this to be positive; this is not the expectation value of a positive operator on $\H \otimes \H^*$!

A particular example of negativity is furnished by $\H = \comps^2$, $\rho$ maximally mixed, and
\begin{align}
	a_1  = \begin{pmatrix}1 & 0 \\ 0 & 0 \end{pmatrix},
	&& a_2 = \begin{pmatrix} 0 & 0 \\ 0 & 1 \end{pmatrix},
	&& a_3 = \begin{pmatrix} 0 & -1 \\ 0 & 0 \end{pmatrix},
	&& a_4 = \begin{pmatrix} 0 & 0 \\ -1 & 0 \end{pmatrix}, \\
	b_1  = \begin{pmatrix}0 & 0 \\ 0 & 1 \end{pmatrix},
	&& b_2 = \begin{pmatrix} 1 & 0 \\ 0 & 0 \end{pmatrix},
	&& b_3 = \begin{pmatrix} 0 & 0 \\ 1 & 0 \end{pmatrix},
	&& b_4 = \begin{pmatrix} 0 & 1 \\ 0 & 0 \end{pmatrix},
\end{align}
for which one can compute equation \eqref{eq:negative-sum} and find the negative answer $(-2).$

\bibliographystyle{JHEP}
\bibliography{bibliography}

\end{document}